\documentclass[a4paper,UKenglish,cleveref, autoref, numberwithinsect]{lipics-v2019}
\usepackage{boxedminipage}
\usepackage{mathtools}
\usepackage{thmtools}
\usepackage{tikz}
\usepackage[textwidth=22mm,textsize=tiny,disable]{todonotes}

\usetikzlibrary{decorations}
\usetikzlibrary{decorations.pathreplacing}
\usetikzlibrary{decorations.pathmorphing}

\newcommand{\pbDef}[3]{% \noindent
    \begin{center}
        \begin{boxedminipage}{\textwidth}
            #1
            \smallskip\\
            \begin{tabular}{lp{0.99 \textwidth - \widthof{~~~Question: } - 2.6pt}}
                Input: & #2\\
                Question: & #3
            \end{tabular}
        \end{boxedminipage}
    \end{center}
}

\DeclarePairedDelimiter\abs{\lvert}{\rvert}
\DeclarePairedDelimiter{\ceil}{\lceil}{\rceil}
\DeclarePairedDelimiter{\floor}{\lfloor}{\rfloor}

\title{Minimizing and Computing the Inverse Geodesic Length on Trees} %TODO Please add

\author{Serge Gaspers}{UNSW Sydney, Australia; and \and Data61, CSIRO, Sydney, Australia}{sergeg@cse.unsw.edu.au}{https://orcid.org/0000-0002-6947-9238}{Serge Gaspers is the recipient of an Australian Research Council (ARC) Future Fellowship (FT140100048).}

\author{Joshua Lau\footnote{Corresponding author}}{UNSW Sydney, Australia}{joshua.lau@unsw.edu.au}{https://orcid.org/0000-0001-7490-633X}{}

\authorrunning{S. Gaspers and J. Lau}%TODO mandatory. First: Use abbreviated first/middle names. Second (only in severe cases): Use first author plus 'et al.'

\Copyright{Serge Gaspers and Joshua Lau}%TODO mandatory, please use full first names. LIPIcs license is "CC-BY";  http://creativecommons.org/licenses/by/3.0/

\ccsdesc[100]{Mathematics of computing~Trees}
\ccsdesc[100]{Mathematics of computing~Graph algorithms}
\keywords{Trees, Treewidth, Fixed-Parameter Tractability, Inverse Geodesic Length, Vertex deletion, Polynomial
multiplication, Distance distribution}%TODO mandatory; please add comma-separated list of keywords

\category{}%optional, e.g. invited paper

\supplement{}%optional, e.g. related research data, source code, ... hosted on a repository like zenodo, figshare, GitHub, ...

\acknowledgements{We thank David Harvey and Ray Li for fruitful discussions and feedback.}
\nolinenumbers %uncomment to disable line numbering

\hideLIPIcs  %uncomment to remove references to LIPIcs series (logo, DOI, ...), e.g. when preparing a pre-final version to be uploaded to arXiv or another public repository

\EventEditors{Pinyan Lu and Guochuan Zhang}
\EventNoEds{2}
\EventLongTitle{30th International Symposium on Algorithms and Computation (ISAAC 2019)}
\EventShortTitle{ISAAC 2019}
\EventAcronym{ISAAC}
\EventYear{2019}
\EventDate{December 8--11, 2019}
\EventLocation{Shanghai University of Finance and Economics, Shanghai, China}
\EventLogo{}
\SeriesVolume{149}
\ArticleNo{62}
\begin{document}

\maketitle

\begin{abstract}

%\jl{Need to rewrite with MinH ideas, e.g. FPT, define MinH, etc}
For any fixed measure $H$ that maps graphs to real numbers, the \textsc{MinH} problem is defined as follows: given a graph $G$,
an integer $k$, and a target $\tau$, is there a set $S$ of $k$ vertices that can be deleted, so that $H(G - S)$ is at
most $\tau$? In this paper, we consider the \textsc{MinH} problem on trees.

We call $H$ \emph{balanced on trees} if, whenever $G$ is a tree, there is an optimal choice of $S$ such that the
components of $G - S$ have sizes
bounded by a polynomial in $n / k$.
We show that \textsc{MinH} on trees is Fixed-Parameter Tractable (FPT) for parameter $n / k$, and
furthermore, can be solved in subexponential time, and polynomial space,
whenever $H$ is additive, balanced on trees, and computable in polynomial time.

%\jl{Am worried the abstract is now too long and too detailed, but am unsure how to shorten it}
A particular measure of interest is the Inverse Geodesic Length (IGL), which is used to gauge the efficiency and connectedness of a graph.
It is defined as the sum of inverse distances between every two vertices: $IGL(G)
= \sum_{\{u,v\} \subseteq V} \frac{1}{d_G(u,v)}$.
% In the \textsc{MinIGL}
% problem, the input is a graph $G$, an integer $k$, and a target inverse
% geodesic length $\tau$, and the question is whether there are $k$ vertices
% whose deletion decreases the IGL of $G$ to at most $\tau$.
%The problem has applications in network analysis and security games.
%
%Aziz et al. (2018) proved that \textsc{MinIGL} is $W[1]$-hard for parameter
%treewidth, and Najeebullah (2018) showed that, under the Exponential Time
%Hypothesis (ETH) of Impagliazzo and Paturi (2001), \textsc{MinIGL} cannot be solved
%in $2^{o(k + n + m)}$ time, even on bipartite graphs.
%SG: Mentioning related work in detail can be deferred to the introduction.
While \textsc{MinIGL} is $W[1]$-hard for parameter treewidth, and cannot be solved in $2^{o(k + n + m)}$ time, even on
bipartite graphs with $n$ vertices and $m$ edges, the complexity status of the problem remains open in the case
where $G$ is a tree. We show that IGL is balanced on trees, to give
a $2^{O\left((n \log n)^{5/6}\right)}$ time, polynomial space algorithm.
% We generalise this result to measures with similar properties, and also give a polynomial time algorithm for the same minimization problem on trees under another related measure, the Wiener index.

The \emph{distance distribution} of $G$ is the sequence $\{a_i\}$ describing
the number of vertex pairs distance $i$ apart in $G$: $a_i = \abs{\{ \{u, v\}:
d_G(u, v) = i\} }$.  Given only the distance distribution, one can easily
determine graph parameters such as diameter, Wiener index, and particularly,
the IGL.  We show that the distance distribution of a tree can be computed in
$O(n \log^2 n)$ time by reduction to polynomial multiplication.  We also extend
the result to graphs with small treewidth by showing that the first $p$ values
of the distance distribution can be computed in $2^{O(\text{tw}(G))} n^{1 +
\varepsilon} \sqrt{p}$ time, and the entire distance distribution can be computed in $2^{O(\text{tw}(G))} n^{1 +
\varepsilon}$ time, when the diameter of $G$ is $O(n^{\varepsilon'})$ for every $\varepsilon' > 0$.

\end{abstract}

\section{Introduction}

The \emph{Inverse Geodesic Length (IGL)} is a widely-used measure for
quantifying the connectedness and efficiency of a given graph or network. In mathematical chemistry, it is
also known as the \emph{Harary Index} \cite{Zhou2008}, and in network science as the \emph{(global)
efficiency} \cite{Crucitti2003}.

To test the resilience of a graph to vertex failures, the problem of minimizing
a particular measure by deleting a fixed number of vertices has been studied
extensively \cite{Morone2015, Kovacs2015, Hossain2013}. In these cases,
heuristics have been used to choose which vertices to delete, and their effect
has been assessed using the chosen measure. In particular, Szczep\'{a}nski et
al.  \cite{Szczepanski2016} chose IGL as the measure to be minimized when
examining this problem. Nonetheless, only recently has the exact optimization
problem itself (\textsc{MinIGL}) been studied.

%\jl{Need to mention the Najeebullah result}
Veremyev et al. \cite{Veremyev2015} formulated \textsc{MinIGL} as a special
case of the Distance-Based Critical Node Detection Problem (DCNP), and
%gave an asymptotically inefficient solution by reduction
reduced the problem
to Integer Linear Programming.
Aziz et al. \cite{ijcai2017-108} observed that \textsc{MinIGL} is NP-complete, since it corresponds to \textsc{Vertex
Cover} when $\tau=0$,
% \sg{$T$ or $\tau$?}
% \jl{$\tau$, thanks!}
but it is also both NP-complete, and $W[1]$-hard for parameter $k$, on both split
and bipartite graphs. Najeebullah \cite{Najeebullah2018} showed that, under the Exponential Time
Hypothesis of Impagliazzo and Paturi \cite{Impagliazzo2001a}, \textsc{MinIGL} cannot be solved in
 $2^{o(k + n + m)}$ time, even on bipartite graphs. On the positive side, it was shown that \textsc{MinIGL} is
Fixed-Parameter Tractable (FPT) for parameter twin (or vertex) cover number,
and also for $\omega + k$, where $\omega$ is the neighbourhood diversity of the
graph. In another paper, Aziz et al.  \cite{aamas2018-gaspers} showed that
\textsc{MinIGL} is $W[1]$-hard for parameter treewidth.
The complexity status of \textsc{MinIGL} when the input graph is a tree was stated as an open question by Aziz et al. \cite{ijcai2017-108,aamas2018-gaspers}, and in open problem sessions of IWOCA 2017 and the Sydney Algorithms Workshop 2017.

In Section~\ref{sec:minigl-trees}, we examine \textsc{MinIGL} on trees, giving the following results.
%showing that it is FPT for parameter $n/k$, and
%use this to provide a subexponential time, polynomial space algorithm.
\begin{restatable}{theorem}{thmfptstatement}
  \label{thm:fptstatement}
  \textsc{MinIGL} is FPT for parameter $n/k$ on trees.
\end{restatable}
\begin{restatable}{theorem}{thmsubexpstatement}
  \label{thm:subexp-statement}
  There is a $2^{O((n \log n)^{5/6})}$ time, $O(n^3)$ space algorithm for
  \textsc{MinIGL} on trees, on a real RAM.
\end{restatable}
\noindent
To do so, we prove more general versions of these results, for the \textsc{MinH} problem in the case when $H$ is
additive, balanced on trees, and computable in polynomial time.

We give a Dynamic Programming (DP) algorithm that solves \textsc{MinIGL}
by matching ordered trees to the structure of the given tree, to give a forest
with $n - k$ vertices and minimum IGL. The running time of this algorithm is
exponential in $L$, but polynomial in $n$, where $L$ is the size of the largest
tree in this forest.
Since $H$ is balanced, $L$ is bounded by a polynomial in $n/k$, so \textsc{MinH} is FPT for parameter $n/k$.
Proving that IGL is balanced on trees then gives \autoref{thm:fptstatement}.
Choosing this DP algorithm when $k$ is large compared to $n$, and a simple brute-force algorithm otherwise, gives
\autoref{thm:subexp-statement}.

% we give an algorithm that tries all subsets
% of $k$ vertices, so it is reasonably efficient when $k$ is relatively small.
% Next,
% Finally, we prove that in any optimal solution, $L$ is
% $O\left(\left(n/k\right)^5\right)$, so the exponential factor in the DP
% algorithm is small when $k$ is relatively large with respect to $n$. Thus, we
% can pick an appropriate threshold (a function of $n$) and elect to use the
% first algorithm when $k$ is below the threshold, and the second algorithm when
% $k$ is above it.  This gives Theorem~\ref{thm:subexp-statement}.
% Our approach is as follows.  First, we give an algorithm that tries all subsets
% of $k$ vertices, so it is reasonably efficient when $k$ is relatively small.
% Next, we give a Dynamic Programming (DP) algorithm that solves \textsc{MinIGL}
% by matching ordered trees to the structure of the given tree, to give a forest
% with $n - k$ vertices and minimum IGL. The running time of this algorithm is
% exponential in $L$, but polynomial in $n$, where $L$ is the size of the largest
% tree in this forest. Finally, we prove that in any optimal solution, $L$ is
% $O\left(\left(n/k\right)^5\right)$, so the exponential factor in the DP
% algorithm is small when $k$ is relatively large with respect to $n$. Thus, we
% can pick an appropriate threshold (a function of $n$) and elect to use the
% first algorithm when $k$ is below the threshold, and the second algorithm when
% $k$ is above it.  This gives Theorem~\ref{thm:subexp-statement}.

\medskip
IGL has been used to identify key protein residues \cite{Bode2007},
compare the robustness of botnet structures \cite{4413000}, and assess the
impact of attacks on power grids \cite{Zhu2014}. Thus, the ability to
compute the IGL of a graph efficiently serves practical purpose in identifying
characteristics of real-world networks.
%, and offers a potential speedup to algorithms for \textsc{MinIGL}.

Since the IGL of a graph can easily be computed from its distance distribution,
we examine the problem of computing the distance distribution of trees.
By combining the relatively well-known techniques of centroid decomposition and
fast polynomial multiplication, we obtain the following result on trees.

\begin{restatable}{theorem}{thmdisttrees}
  \label{thm:dist-trees}

  The distance distribution of a tree with $n$ vertices can be computed in
  $O(n \log^2 n)$ time on a log-RAM.
\end{restatable}

We extend this result to graphs with small treewidth.
This is of practical note, as real-world graphs for which IGL is an indicator
of strength -- such as electrical grids \cite{Atkins2009} and road transport
networks \cite{maniu2018experimental} -- have been found to have relatively
small treewidth.

The distance distribution of a graph can be trivially computed from the All
Pairs Shortest Paths (APSP). The output of APSP is of size $n^2$, so any APSP
algorithm requires $\Omega(n^2)$ time. On graphs with treewidth $k$, APSP can
be computed in $O(kn^2)$ time \cite{Planken2012}, so we seek algorithms that
find the distance distribution with a subquadratic dependence on $n$.
Abboud et al. \cite{Abboud2016} proved that, under the Orthogonal Vectors
Conjecture (OVC), there is no algorithm that distinguishes between graphs of
diameter 2 and 3 in $2^{o(k)} n^{2 - \varepsilon}$ time.  Williams
\cite{Williams2004} showed that the OVC is implied by the Strong Exponential
Time Hypothesis (SETH) of Impagliazzo, Paturi and Zane \cite{Impagliazzo2001a,
Impagliazzo2001}. Since the distance distribution of a graph immediately gives
its diameter, this hardness result also applies to computing the distance
distribution.  We prove the following result.

\begin{restatable}{theorem}{thmdistpref}
  \label{thm:dist-pref}

  The prefix $a_1, \dots, a_p$ of the distance distribution of a graph with $n$
  vertices and treewidth $k$ % was: diameter $p$
  can be computed in $2^{O(k)} n^{1 + \varepsilon} \sqrt{p}$ time on a log-RAM,
  for any $\varepsilon > 0$.
%\sg{Change ``The prefix ... distribution'' to ``The distance distribution'', since we can pad the prefix with $0$'s in linear time to obtain the distance distribution?}
%\jl{If we are only computing a prefix of the distribution, our algorithm is
%  independent of the diameter of the graph. We can change the result to compute
%  the entire distribution when the diameter of the graph is bounded by $p$, but
%I think this result is less general. Thoughts?}
%\sg{Now it makes sense. (Previously, all values up to the diameter had to be computed.)}
\end{restatable}

\noindent
In particular, the number of relevant values of $p$ is at most the graph's diameter, so when the diameter is
$O(n^{\varepsilon'})$ for every $\varepsilon' > 0$, we obtain a $2^{O(k)} n^{1+\varepsilon}$ time algorithm to compute the distance distribution.
This matches the known hardness bounds above, in the sense that under the OVC, (or the stronger SETH), the dependence
on $k$ must be $2^{\Omega(k)}$ when the dependence on $n$ is subquadratic.

Cabello and Knauer \cite{Cabello2009} reduced the problem of computing the
Wiener index \cite{Wiener1947} (the sum of distances between every two vertices) to orthogonal
range queries in $k-1$ dimensions. They did so by applying a divide-and-conquer
strategy that divides the graph with small separators that are found
efficiently.  Abboud et al.  \cite{Abboud2016} adapted this approach to find
radius and diameter.  We take a similar approach, but reduce computing the
distance distribution to the following problem rather than to orthogonal range
queries.

%\jl{I've moved this definition here: feel free to change it if this feels out
%of place / clumsily worded.}
If $v$ and $w$ are vectors in $\mathbb{R}^d$, write $v < w$ if each coordinate of
$w$ is strictly greater than the corresponding coordinate in $v$. In this case,
we say that $w$ (strictly) \emph{dominates} $v$.  We define the
\textsc{RedBluePolynomial} problem as follows.
% \sg{Add a definition for ``strictly dominates''.}

\pbDef{\textsc{RedBluePolynomial}}{$r$ red points $R_1, \dots R_r$, and
  $b$ blue points $B_1, \dots, B_b$ in $\mathbb{R}^d$, along with corresponding
  non-negative integer values $r_1, \dots, r_r$, and $b_1, \dots, b_b$,
  respectively. }{Determine the non-zero coefficients of the polynomial
  $\sum_{(p, q): R_p < B_q} x^{r_p + b_q}$, as a list of (exponent,
  coefficient) pairs.}

%  \sg{What is $V$?}
%  \jl{Changed it to $M$ in Section 4, although it's only important for our
%  solutions.} \sg{Good, it should not be part of the problem definition then.}
% \sg{Can we get rid of ``(at most)''?} \sg{I suggest to move the definition of
% ``strictly dominates'' before the problem definition.}
% \sg{Can we use $r$ and $b$ instead of $\mathcal{R}$ and $\mathcal{B}$? Lower-case letters are more common for numbers, whereas calligraphic upper-case letters are more common for sets.}
% \jl{Yes, this is a good idea - thanks!}

\noindent
This problem can be solved naively in quadratic time, but we seek a more efficient % was: optimal
solution in the case when the value of each point is bounded.

To our knowledge, this problem is new, and a variant of a well-known counting
problem, which asks for the number of red points dominated by each blue point.
Chan and P\v{a}tra\c{s}cu \cite{Chan2010} showed that this variant can
be solved in $O(n \sqrt{\log n})$ time on a Word RAM, using word operations to
facilitate efficient counting. Bentley \cite{Bentley1980} gave a
multidimensional divide-and-conquer approach for a similar problem, which
Monier \cite{Monier1980} showed had complexity $O(dn \cdot B(n, d))$ where
$B(n, d) = \binom{d + \ceil{\log n}}{d}$.

Bringmann et al. \cite{Bringmann2018} used this fact to show that the method
employed by Cabello and Knauer \cite{Cabello2009}, and Abboud et al.
\cite{Abboud2016} can, in fact, be used to compute the Wiener index, radius,
and diameter of graphs with treewidth $k$ in $2^{O(k)} n^{1 + \varepsilon}$
time for any $\varepsilon > 0$, by proving that $B(n, k) = 2^{O(k)}
n^{\varepsilon}$. Furthermore, Husfeldt \cite{Husfeldt2017} gave an improved
$2^{O(k)} n$ time algorithm for computing diameter and radius in the case where
the graph also has constant diameter.  However, it was noted that this result
only pertains to the existence of pairs of vertices at certain distances, and
not to counting the number of such pairs.  Thus, the result does not directly
give further insight to computing distance distributions.

We follow Bentley's method, where it suffices to consider the one-dimensional
case, $d = 1$. We resolve this case using square-root decomposition and fast
polynomial multiplication. Applying the approach of Bringmann et al. to analyse
the running time of this approach gives Theorem~\ref{thm:dist-pref}. A detailed
discussion of this algorithm is given in Section~\ref{sec:compute-igl}.

\section{Preliminaries}

%\subsection{Definitions}

Let $G = (V, E)$ be a graph and suppose $u, v, w \in V$. We define the
\emph{distance} $d_G(u, v)$ between $u$ and $v$ to be the fewest number of
edges in any path from $u$ to $v$, or $\infty$ if no such path exists, with the convention that $\frac{1}{\infty} = 0$.

% The \emph{Inverse Geodesic Length} or \emph{IGL} of $G$ is $IGL(G) =
% \sum_{\{u,v\} \subseteq V} \frac{1}{d_G(u,v)}$, with the convention that
% $\frac{1}{\infty} = 0$. The \textsc{MinIGL} problem is defined as follows.

%\pbDef{\textsc{MinIGL}}{a graph $G=(V,E)$, an integer $k$, and a rational
%number $Q$.}{Is there a vertex subset $S\subseteq V$ with $|S|\le k$ such that $IGL(G-S) \le Q$?}

%\noindent
In Section~\ref{sec:minigl-trees}, we consider the problem when the provided
graph is a tree $T$. In this case, precisely one simple path exists between
every pair $\{u, v\} \subseteq V$. Define $\mathcal{P}_T(u, v)$ to be the set of
vertices along the simple path from $u$ to $v$ in $T$, including the endpoints
$u$ and $v$.

Observe that $d_T(u, w) + d_T(w, v) =
d_T(u, v)$ if and only if $w \in \mathcal{P}_T(u, v)$.
For a vertex $w$, we also define $\mathcal{P}^{-1}_T(w)$ to be the set of all
(unordered) pairs of vertices whose path in $T$ passes through $w$.
Formally, $\mathcal{P}^{-1}_T(w) = \{ \{u, v\} \subseteq V : w \in \mathcal{P}_T(u, v) \}$.

A vertex $u$ is a \emph{centroid} of $T$ if the maximum
size of a connected component in $T - u$ is minimized.
%The following lemma characterises the existence of centroids in a tree.
We will use the following results, concerning centroids.

\begin{lemma}[Jordan \cite{Jordan1869}] % straight brackets, not curly ones
  \label{lem:centroid-half}
  Every tree has either one centroid or two adjacent centroids. If a centroid is
  deleted from a tree, each tree in the remaining forest contains no more than
  $\frac{n}{2}$ vertices, where $n$ is the number of vertices in the original
  tree.
\end{lemma}
% We can also see that the number of pairs of vertices disconnected by the removal of a centroid is quadratic in the
% number of vertices in the tree.
\begin{restatable}{lemma}{lemcentroiddisconnect}
  \label{centroid-disconnect}
  Let $u$ be a centroid of a tree $T$ with $n \geq 2$ vertices.
  Then, $\abs{\mathcal{P}^{-1}_T(u)} \geq \frac{n^2}{4}$.
\end{restatable}

\begin{proof}
  The statement is trivially true for the only tree with $n = 2$ vertices.
  Otherwise, root the tree at $u$, and let the children of $u$ be $v_1, \dots,
  v_{c(u)}$.  Note that $c(u) \geq 2$ when $n \geq 3$: if $c(u) = 1$ then $T -
  u$ has a single component with $n - 1 > \frac{n}{2}$ vertices, so $u$ would
  not be a centroid of $T$, according to Lemma~\ref{lem:centroid-half}.

  Suppose the subtree rooted at $v_i$ has $a_i$ vertices, so $a_1 + \dots +
  a_{c(u)} = n - 1$. Now $\abs{\mathcal{P}^{-1}_T(u)}$ is the number of paths in $T$
  passing through $u$. We can view each path through $u$ as starting in some
  subtree $v_i$ and ending at either $u$ or some other subtree $v_j$ where $j
  \neq i$. Hence, the number of paths through $u$ can be determined only from
  these subtree sizes, so we need only consider the distributions of the number
  of vertices among the subtrees.
  
  Without loss of generality, further assume that $a_1 \geq a_2 \geq \dots \geq
  a_{c(u)} \geq 1$. We will show that $\abs{\mathcal{P}^{-1}_T(u)}$ is minimized
  when $c(u) = 2$.
 
  Suppose, for a contradiction, that $c'(u) \geq 3$ and $a_1' \geq a_2' \geq
  \dots \geq a'_{c'(u)} \geq 1$ is some sequence of subtree sizes that induces
  strictly fewer paths through $u$. We also have that $a_{c'(u)} \geq 1$.
  Further, we must have that either $a_1' < \ceil{\frac{n-1}{2}}$, or $a_1' = \ceil{\frac{n-1}{2}}$ and $a_2' <
  \floor{\frac{n-1}{2}}$.
  
  Now suppose we move one vertex from the smallest subtree to the largest (or the second largest, if $a_1' =
  \ceil{\frac{n-1}{2}}$) subtree. In both cases, the number of vertices in every subtree remains at most
  $\ceil{\frac{n-1}{2}}$. As a result, $u$ must remain a centroid: the deletion of a vertex in one such subtree will leave a component containing all vertices in
  the other subtrees plus the root, which together must have size at least $\ceil{\frac{n-1}{2}}$.
  Next, consider how $\abs{\mathcal{P}^{-1}_T(u)}$ has changed: we have invalidated at least
  $a_2'$ paths and created only $a'_{c'(u)}-1$ additional paths, so we have obtained a
  scenario with strictly fewer paths through a centroid, which is a contradiction.

  Hence, $\abs{\mathcal{P}^{-1}_T(u)}$ is minimized when $c(u) = 2$, so we must have
  $a_1 = \ceil{\frac{n-1}{2}}$ and $a_2 = \floor{\frac{n-1}{2}}$, because $a_1
  \leq \frac{n}{2}$ by Lemma~\ref{lem:centroid-half}. Therefore, there are
  $(a_1 + 1)(a_2 + 1) - 1$ paths through $u$, accounting also for those paths with
  $u$ as an endpoint. We have
  \begin{align*}
    (a_1 + 1)(a_2 + 1) - 1 &= \ceil*{\frac{n+1}{2}} \floor*{\frac{n+1}{2}} - 1 \\
                           &\geq \frac{n}{2}\left(\frac{n}{2} + 1\right) - 1 \\
                           &\geq \frac{n^2}{4}
  \end{align*}
  whenever $n \geq 2$, as required.
\end{proof}

%\jl{Mention proof omitted?}

%\jl{Do we need to restate the distance distribution definition / problem here?} Yes
%\sg{restate the distance distribution definition / problem here}

In Section~\ref{sec:compute-igl} we also consider the problem of computing the IGL, using the tree decompositions of
graphs with small treewidth.
%We use
%common definitions of tree decompositions, nice tree decompositions and treewidth, which are included in
%\autoref{sec:missing-defns}, for completeness.
%\jl{Don't point to appendix?}
A \emph{tree decomposition} of $G$ is a tree $\mathcal{T}$ whose vertices (called \emph{nodes}) are $\{1, \dots, I\}$ and a sequence $\mathcal{V}_1, \dots, \mathcal{V}_I$ of subsets of $V$ (called \emph{bags}) such that

\begin{enumerate}
  \item $V = \bigcup_{i=1}^I \mathcal{V}_i$;
  \item If $uv \in E$, then $\{u, v\} \subseteq \mathcal{V}_i$ for some $i$;
  \item $\mathcal{V}_a \cap \mathcal{V}_c \subseteq \mathcal{V}_b$ whenever $b \in
    \mathcal{P}_\mathcal{T}(a, c)$.
\end{enumerate}

\noindent
The \emph{width} of such a tree decomposition is $\max_{i=1}^I \abs{\mathcal{V}_i} -
1$.  The \emph{treewidth} $\text{tw}(G)$ of $G$ is the minimum width among all tree
decompositions of $G$.

%Further, if $\mathcal{T}$ is a rooted binary tree, we say
%that $(\mathcal{T}, \{\mathcal{V}_i\})$ is a \emph{nice} tree decomposition if each $i
%\in I$ also satisfies the following conditions:
%
%\begin{enumerate}
%  \item If $i$ is a leaf of $\mathcal{T}$, and not the root, then $\abs{\mathcal{V}_i} = 1$;
%  \item If $i$ has a single child $j$ in $\mathcal{T}$, then either $\abs{\mathcal{V}_i}
%    = \abs{\mathcal{V}_j} + 1$ and $\mathcal{V}_j \subset \mathcal{V}_i$, or, $\abs{\mathcal{V}_i} =
%    \abs{\mathcal{V}_j} - 1$ and $\mathcal{V}_i \subset \mathcal{V}_j$;
%  \item If $i$ has two children $i_1$ and $i_2$ in $\mathcal{T}$, then $\mathcal{V}_i =
%    \mathcal{V}_{i_1} = \mathcal{V}_{i_2}$.
%\end{enumerate}

% The \emph{distance distribution} of $G$ is the sequence $\{a_i\}$ which
% describes the number of pairs of vertices in $G$ that are distance $i$ apart.
% Formally, $a_i = \abs{\{ \{u, v\}: d_G(u, v) = i\} }$. This definition extends
% naturally to directed graphs (counting instead, the number of ordered
% pairs) and graphs with edges whose weights are non-negative integers. In
% Section~\ref{sec:compute-igl}, we examine the following problem on trees, and
% graphs with bounded treewidth.
% 
% \pbDef{\textsc{DistanceDistributionPrefix}}{a graph $G=(V,E)$ and an integer
% $P$.}{What are the values $a_1, \dots, a_p$ of the distance distribution of
% $G$?}
% 
% \noindent
% Particularly, computing the entire distance distribution of the graph
% corresponds to the case when $P = \abs{V}-1$.

\subsection{Model of computation}

We establish our results on models of computation that closely reflect what is
available to programmers of high-level languages on physical computing devices
today.

In Section~\ref{sec:minigl-trees}, we solve \textsc{MinIGL} by explicitly
computing the minimum IGL that can be obtained by deleting $k$ vertices from
the given tree. We perform this on the \emph{real RAM} formulated by Shamos
\cite{Shamos1978}, which allows addition, subtraction, multiplication,
division and comparisons of real numbers in constant time, but does not support
rounding a value to the nearest integer, or modulo as native operations. 
This permits efficiently adding and comparing contributions of distances
between vertices to the IGL.

In Section~\ref{sec:compute-igl}, we reduce the problem of computing the IGL of
a graph to finding its distance distribution. We solve this on a \emph{log-RAM}
introduced by F\"{u}rer \cite{Furer2014}, which is a Word RAM that also
supports constant time arithmetic operations (including multiplication, integer
modulo, and division) on words of length $O(\log n)$. F\"{u}rer showed that on
a log-RAM, multiplication of two $n$-bit integers can be done in $O(n)$ time,
using either the approach of Sch\"{o}nhage and Strassen \cite{Schonhage1971}
(performing a complex polynomial-based Fast Fourier Transform (FFT) and
maintaining sufficient precision), or that of F\"{u}rer \cite{Furer2009}
(performing an FFT over a ring of polynomials).

We extend this to integer polynomials with bounded coefficients, as follows.

\begin{restatable}{lemma}{lemfastmult}
  \label{lem:fast-mult}
  Suppose $P$ and $Q$ are integer polynomials of degree $n$ whose coefficients
  are non-negative integers, such that their product $PQ$ has coefficients not
  exceeding some integer $m$. Then, the coefficients of $PQ$ can be
  computed from the coefficients of $P$ and $Q$ in $O(n \log m)$ time on a
  log-RAM.
\end{restatable}

\begin{proof}
  This can be done using Kronecker substitution
  \cite{kronecker1882grundzuge}. Let $m'$ be the smallest power of two strictly
  greater than $m$. Evaluate $P(m')$ and $Q(m')$ to give two integers, each
  containing $O(n \log m)$ bits. This is easily done through bit manipulation
  because $m'$ is a power of two. Next, multiply these together to obtain the
  value of $(PQ)(m')$: this takes time linear in the number of bits in the
  multiplicands. Finally, one can unpack the coefficients of $PQ$ from this
  value in $O(n \log m)$ time by viewing it as a base $m'$ integer: since $m'$
  is a power of two, this can be simply read off the bits of the result.
\end{proof}

%
%In Section~\ref{sec:minigl-trees}, we solve \textsc{MinIGL} on the \emph{real RAM} formulated by Shamos
%\cite{Shamos1978} which allows us to efficiently add and compare contributions of distances between vertices to the IGL.
%
%In Section~\ref{sec:compute-igl}, we compute distance distributions on the \emph{log-RAM} introduced by F\"{u}rer
%\cite{Furer2014}.
%There, we reduce computing the IGL on graphs with bounded treewidth to the multiplication of integer polynomials, which can be
%done efficiently under this model of computation.
%This is proven as \autoref{lem:fast-mult}.
%
%The rationale for choosing these models can be found in Appendix~\ref{appendix:moc}.

\section{MinIGL on Trees}

\label{sec:minigl-trees}

In this section,  we give a new subexponential time, polynomial space algorithm for \textsc{MinH} on trees, when $H$
satisfies the following properties.
We use this to prove Theorems~\ref{thm:fptstatement} and \ref{thm:subexp-statement}, by showing that IGL also satisfies
these properties.

\begin{definition}[Additivity]
    We say that a measure $H$ on graphs is \emph{additive} if $H(G_1 \oplus G_2) = H(G_1) + H(G_2)$
    for any vertex-disjoint graphs $G_1 = (V_1, E_1)$ and $G_2 = (V_2, E_2)$, 
    where $G_1 \oplus G_2$ is the graph $(V_1\ \dot\cup\ V_2, E_1\ \dot\cup\ E_2)$.
\end{definition}

Call a forest \emph{$L$-trimmed} if none of its trees contain more than $L$
vertices.  In the same way, call a subset of vertices in a tree
\emph{$L$-trimming} if their deletion gives an $L$-trimmed forest.

\begin{definition}[Balanced on trees]
    \label{def:balance}
    We say that a measure $H$ is \emph{balanced (on trees)} if
    there exist positive constants $c_H$ and $t_H$, such that, for
    any \textsc{Yes}-instance $(T, k, \tau)$ of \textsc{MinH} on a tree $T$
    with $n$ vertices, there exists a witness that is $c_H
    \left(n/k\right)^{t_H}$-trimming. 
\end{definition}

Hereafter, we will assume that the value of $H$ on a forest is computable in $O(n^{\alpha})$ time, and $O(n^{\beta})$
space, on a real RAM, where $\alpha, \beta \geq 1$ are constants.
We also assume that such a value can be stored in a constant number of words on a real RAM.

We prove Theorems~\ref{thm:fptstatement} and \ref{thm:subexp-statement} by giving compatible algorithms for
\textsc{MinH} on trees, then complete the proof by showing that IGL satisfies the same properties that $H$ does.
Now it is easy to see that there is a na\"{i}ve, brute-force algorithm for \textsc{MinH}.

\begin{restatable}{lemma}{lemnaive}
  \label{lem:bf-h-tree}
  There is an $O(n^{k+\alpha})$ time, $O(n^{\beta})$ space algorithm for
  \textsc{MinH} on a tree, on a real RAM.
\end{restatable}

\begin{proof}
%  \jl{I've changed this back to not use the result from the ComputeIGL section:
%    it's not clear that the algorithm on the log-RAM directly translates to the
%    real RAM: perhaps with an extension that allows modulus and rounding on
%  word-sized integers, but this is not native.} %SG: Ok
  We simply try all $\binom{n}{k} = O(n^k)$ subsets of $k$ vertices. The value
  of $H$ on the forest that remains after each subset has been removed can be
  computed in $O(n^{\alpha})$ time and $O(n^{\beta})$ space.
\end{proof}

If $k$ is small, this algorithm may be efficient. When $k$ is large, the
vertices forming an optimal solution will leave a forest of relatively small
trees after they are deleted, since $H$ is balanced. We use this property to
develop an alternate, more efficient algorithm for \textsc{MinH} in this case.
Let $L = c_H \left(n/k\right)^{t_H}$. 
Our algorithm minimizes $H$, considering only $L$-trimming subsets of $k$ vertices.
%: it suffices to consider these, since $H$ is balanced on trees.
The running time of this algorithm is exponential in $L$, but polynomial in $n$, so it is fast when $k$ is large,
relative to $n$. 

\begin{lemma}
  \label{lem:expl-h-tree}

  Let $T = (V, E)$ be a tree with $n$ vertices. There is an $O \left(
  \frac{4^L}{\sqrt{L}} \left(n^2 + L^{\alpha - 1}\right) \right)$ time,
  $O\left(nkL + L^{\beta}\right)$ space algorithm on a real RAM, which finds
  the minimum value of $H(T - S)$, among all $L$-trimming subsets $S$ of $k$
  vertices.
\end{lemma}

%\begin{proof}[Proof (Algorithm~\ref{lem:expl-h-tree})]
\begin{proof}
  We root $T$ arbitrarily and employ DP to compute
  this minimum value for every subtree and budget, in two cases: the case where the root
  of the subtree is deleted, and the case where it is not. Denote these minimum
  values by $f(u, b)$ and $g(u, b)$, respectively, for the subtree rooted at
  $u$ and budget $b$. The leaves of the tree form the base cases for this
  algorithm, and the final answer is derived from the minimum of $f(\text{root}, k)$
  and $g(\text{root}, k)$. It remains to give recurrences for $f$ and $g$.

  In the case where $u$ is deleted, we simply need to distribute the remaining
  $b-1$ deletions among the subtrees rooted at each child of $u$. Let the
  children of $u$ be $v_1, \dots, v_{ch_T(u)}$ in a fixed order. Our recurrence takes the form of another
  DP algorithm: let $f'(u, i, b')$ be the minimum value of $f$
  distributing a budget of $b'$ deletions among the subtrees rooted at the
  first $i$ children of $u$. Our recurrence is as follows:
  \[
    f'(u, i, b') = \min_{0 \leq b'' \leq b'} ( \min(f(v_i, b''), g(v_i, b''))
    + f'(u, i-1, b' - b'') )
  \]
  \noindent
  and we have that $f(u, b) = f'(u, ch_T(u), b-1)$.

  If $u$ is not deleted, it will be the root of some tree with no more than $L$
  vertices after our chosen subset has been deleted. We fix the structure (formally, an
  \emph{ordered tree}) for this rooted tree, and attempt to match the vertices
  in this structure to vertices in the subtree rooted at $u$. Formally, let the
  structure be an ordered tree $T'$ over $L' \leq L$ vertices. Let its vertex
  set be $V' = \{1, \dots, L'\}$ and, without loss of generality, suppose $1$
  is its root. We seek a total, injective mapping $m: V' \rightarrow V$
  satisfying the following conditions. 
  \begin{enumerate}
    \item $m(1) = u$;
    \item Suppose $p$ and $p'$ are the parents of $q$ and $q'$ in $T$ and $T'$,
      respectively. If $m(q') = q$ then $m(p') = p$;
    \item Let $p$ and $p'$ be vertices in $T$ and $T'$ such that their children
      are, in order, $q_1, \dots, q_{ch_T(p)}$ and $q'_1, \dots,
      q'_{ch_{T'}(p')}$, respectively. If $m(q'_{j_1}) = q_{i_1}$, $m(q'_{j_2})
      = q_{i_2}$ and $j_1 \leq j_2$, then $i_1 \leq i_2$. That is, children are
      matched in order.
  \end{enumerate}

  % \noindent
  % Figure~\ref{fig:ordered-mapping} gives an example of a valid such mapping.

  \begin{figure}
    \tikzstyle{vertex}=[circle,draw=black,thick,minimum size=9pt,inner sep=0pt]
    \tikzstyle{selected vertex} = [vertex, fill=red!24]
    \tikzstyle{edge} = [draw,thick,-]
    \tikzstyle{weight} = [font=\small]

    \begin{center}
      \begin{tikzpicture}[node distance=1cm,scale=0.6]
        \usetikzlibrary{positioning}

        \draw node[] () at (0, 4.7) {$T'$};
        \draw node[] () at (6.25, 4.7) {$T$};

        \draw node[vertex] (1) at (0, 4) {\footnotesize 1};
        \draw node[vertex] (2) at (-1, 3) {\footnotesize 2};
        \draw node[vertex] (3) at (1, 3) {\footnotesize 3};
        \draw node[vertex] (4) at (-1, 2) {\footnotesize 4};
        \draw node[vertex] (5) at (0.3, 2) {\footnotesize 5};
        \draw node[vertex] (6) at (1.7, 2) {\footnotesize 6};

        \foreach \source/ \dest in
        {1/2, 1/3, 2/4, 3/5, 3/6}
        \path[edge] (\source) -- node[weight] {} (\dest);
        % \draw[edge] (i) -- node[weight] {} (-0.5, -2.5);
        % \draw[edge] (i) -- node[weight] {} (0.5, -2.5);

        % \draw[edge] (j) -- node[weight] {} (2, -2.5);

        % \draw[edge] (k) -- node[weight] {} (3.3, -2.5);
        % \draw[edge] (k) -- node[weight] {} (4, -2.5);
        % \draw[edge] (k) -- node[weight] {} (4.7, -2.5);

        % \draw[edge] (r) -- node[weight] {} (4.5, 6.9);
        % \draw[edge, dashed] (r) -- node[weight] {} (5, 5.4);

        \draw node[vertex] (a) at (6.25, 4) {\footnotesize 1};
        \draw node[vertex, fill=gray!30] (b) at (4, 3) {};
        \draw node[vertex] (c) at (5.5, 3) {\footnotesize 2};
        \draw node[vertex, fill=gray!30] (d) at (7, 3) {};
        \draw node[vertex] (e) at (8.5, 3) {\footnotesize 3};

        \draw node[vertex, fill=gray!30] (f) at (5, 2) {};
        \draw node[vertex] (g) at (6, 2) {\footnotesize 4};

        \draw node[vertex] (h) at (7.8, 2) {\footnotesize 5};
        \draw node[vertex, fill=gray!30] (i) at (8.5, 2) {};
        \draw node[vertex] (j) at (9.2, 2) {\footnotesize 6};

        \draw node[vertex] (k) at (4, 2) {};

        \draw node[vertex] (l) at (4.6, 1) {};
        \draw node[vertex] (m) at (5.4, 1) {};

        \draw node[vertex] (n) at (7, 2) {};

        \foreach \source/ \dest in
        {a/b, a/c, a/d, a/e, c/f, c/g, e/h, e/i, e/j, b/k, f/l, f/m, d/n}
        \path[edge] (\source) -- node[weight] {} (\dest);

        \path[edge, dashed, ->] (1) -> node[weight, above] {$m$} (a);
      \end{tikzpicture}

      \caption[\textsc{MinH} on trees, Dynamic Programming. Mapping the vertices of ordered tree $T'$ to $T$]{Mapping the vertices of $T'$ to $T$ in
        Lemma~\ref{lem:expl-h-tree}. Note that $T'$ is an ordered tree, and
        that children (and their subtrees) must be mapped in order. Shaded
        vertices will be deleted, and we recursively solve for the subtrees
      rooted at their children.}
    \label{fig:ordered-mapping}
    \end{center}
  \end{figure}
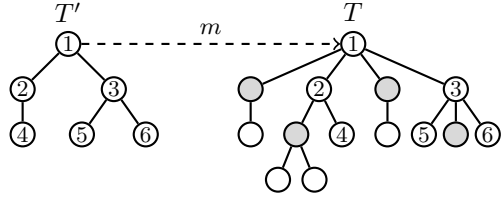

  \noindent
  Note that the structure of the chosen ordered tree uniquely characterises the
  value of $H$ on the component containing $u$, since $H$ is only defined on
  unlabelled graphs, and is additive, so this value is independent of the
  structure of other components.

  Let $v$ be some vertex in $T$. If $v$ is mapped to by $m$, then $v$ is a part
  of this component. Otherwise, if $v$ is not mapped to by $m$ but its parent
  is, then $v$ must be a vertex chosen for deletion, and so we should
  recursively consider each of its childrens' subtrees.

  This implies a DP approach to determine the optimal choice of $m$, similar to
  that of $f'$. We let $g'(u, i, b', u', j)$ be the minimum value (of $H$)
  induced by a mapping which maps $u'$ to $u$ and maps the first $j$ children
  of $u'$ among the first $i$ children of $u$ with a total budget of $b'$
  deletions in the subtree rooted at $u$. This value \emph{does not} include the contributions of vertex pairs
  in $T$ which both end up in the current component (are mapped to by $m$).
  
  We have a choice to either delete the
  $i$th child $v_i$ of $u$, or map it to the $j$th child $v'_j$ of $u'$. In
  both cases, we allocate a budget of $b'' \leq b'$ deletions to the subtree
  rooted at $v_i$.  This gives the following recurrence:
  \[
    g'(u, i, b', u', j) = \min_{0 \leq b'' \leq b'} \min \begin{cases}
       g'(u, i-1, b' - b'', u', j) + f(v_i, b''), \\
       g'(u, i-1, b' - b'', u', j-1) + g'(v_i, ch_T(v_i), b'', v'_j, ch_{T'}(v'_j)) 
    \end{cases}
  \]
  and $g(u, b) = \min_{T'} (H(T') + g'(u, ch_T(u), b, 1, ch_{T'}(1)))$.
  This concludes the description of the algorithm.

  We now prove the time and space complexity of this algorithm.
  First, consider the complexity of computing $f'$ (and thus that of computing
  $f$) over the entire tree. Each $(u, i)$ pair corresponds to an edge of the
  tree, and $b' \leq k$, so there are $O(nk)$ states of $f$ and $f'$ to
  consider. Further, if $b'$ is no less than the number of vertices in the
  subtree rooted at $u$, we may delete every vertex in that subtree, giving a
  minimum IGL of 0. Hence, we need only consider values for $b''$ which do
  not exceed the number of vertices in the subtree rooted at $v_i$, and also
  those where $b' - b''$ does not exceed the total number of vertices in the
  subtrees rooted at $v_1, \dots, v_{i-1}$. It can be shown (see, e.g.,
  \cite{Cygnan2012}) that there are only $O(n^2)$ $( (u, i, b'), b'')$ tuples
  satisfying these properties in the entire tree.
  Since each (state, recurrence) pair corresponds to one such tuple, all values of $f$ can
  be computed in $O(n^2)$ time and $O(nk)$ space, assuming the required values
  of $g$ are readily available.

  Similarly, for a fixed $T'$ we can examine the complexity of computing the
  values of $g'$. Each $(u', j)$ pair corresponds to an edge of $T'$, and there
  are at most $L - 1$ of these. By a similar argument to that for $f'$, we only
  need to consider $O(n^2)$ tuples $(u, i, b', b'')$. Hence, we can compute all
  values of $g'$ in $O(n^2 L)$ time and $O(nkL)$ space.

  It is a well-known result that the number of ordered trees with $l$ vertices
  is $C_{l-1}$, the $(l-1)$-th Catalan Number (see, for example,
  \cite{Flajolet:2009:AC:1506267}). The sum of the first $l$ Catalan numbers
  was upper bounded by $\frac{4^{l+1}}{3 \sqrt{\pi l^3}}$ by
  Topley \cite{Topley2016}, so there are $O\left(\frac{4^L}{L \sqrt{L}} \right)$
  different ordered trees we need to try as $T'$. We generate these
  recursively. For each in turn, we compute its value of $H$ in $O(L^{\alpha})$
  time and $O(L^{\beta})$ space. Then, we compute $g'$, which dominates our
  running time and memory consumption.  Thus, our algorithm runs in $O\left(
  \frac{4^L}{\sqrt{L}} (n^2 + L^{\alpha - 1}) \right)$ time and uses $O(nkL +
  L^{\beta})$ space, as required.
\end{proof}

Since $n^2 + L^{\alpha - 1} = O\left(n^{\max(2, \alpha - 1)}\right)$, and $L
\leq c_H \left(\frac{n}{k}\right)^{t_H}$, it follows that \textsc{MinH} is FPT
for parameter $\frac{n}{k}$.

\begin{restatable}{corollary}{corminhfpt}
    \label{thm:minh-fpt}
    Suppose $H$ is a measure on graphs, that is additive, balanced on trees,
    and computable in polynomial time on trees, on a real RAM. Then
    \textsc{MinH} is FPT for parameter $n/k$ on trees.
\end{restatable}

% We prove this by proposing two separate algorithms, choosing which to use
% depending on whether $k$ is small or large, relative to $n$. Both algorithms
% solve \textsc{MinH} by computing the minimum possible value of $H$ that can be
% obtained by deleting a set $S$ of $k$ vertices from the input graph. We call
% $S$ an \emph{optimal subset}, if removing $S$ from the input graph gives this
% minimum value, and our algorithm returns that this is a \textsc{Yes}-instance
% if and only if this minimum value is at most $\tau$. 
With an appropriate threshold, we can combine the approaches of
\autoref{lem:bf-h-tree} and \autoref{lem:expl-h-tree} to give a subexponential time, polynomial space
algorithm for \textsc{MinH}.

\begin{restatable}{corollary}{corsubexpminh}
  \label{cor:minh-subexp}
  Suppose $H$ is a measure on graphs, that is additive, balanced on trees, and computable in polynomial time on
  trees, on a real RAM.
  Then there is a $2^{O\left(\left(n \log n\right)^{t_H/\left(t_H + 1\right)}\right)}$ time, polynomial space algorithm
  for \textsc{MinH} on trees, where $t_H$ is the constant given in \autoref{def:balance}.
\end{restatable}

\begin{proof}
  \autoref{lem:bf-h-tree} gives us an $O(n^{k+\alpha}) = 2^{O(k \log n)}$ time  
  algorithm for \textsc{MinH} on a tree. \autoref{lem:expl-h-tree} gives us an
  alternate $O\left(4^{c_H \left( n/k \right)^{t_H}} n^{\max(2,
          \alpha-1)}\right) = 2^{O\left(\left( n/k \right)^{t_H} + \log
  n\right)}$ time algorithm for the same problem. Note that the memory
  consumption of both algorithms is bound by $O(n^{\max(3, \beta)})$, so they
  are both polynomial in space.
  
  Let $k^{\ast} = n^{t_H/(t_H + 1)} \log^{-1/(t_H + 1)} n$. We select the
  former algorithm when $k \leq k^{\ast}$, and the latter algorithm otherwise.
  In both cases, our running time is bound by $2^{O\left(\left(n \log
              n\right)^{t_H/\left(t_H + 1\right)}\right)}$, as required.
\end{proof}

%\jl{I removed reference to diameter-balanced on trees}
% We are also aware of an $3^L \cdot \text{poly}(n)$ time DP algorithm
% operating over the preorder traversal of the tree (once rooted). The state
% encodes a current vertex $u$, and, for each ancestor $a$ of the current vertex
% and for each distance $l$, the number of vertices in the subtree rooted at $a$
% that are distance $l$ from $u$ in the final forest. We elect to use the result
% of Lemma~\ref{lem:expl-h-tree} instead, as its exponent is larger only by a
% constant, whereas its memory consumption is strictly polynomial.

We now prove that IGL satisfies the requirements of \autoref{thm:minh-fpt} and \autoref{cor:minh-subexp}.
IGL is clearly additive, since pairs of vertices
belonging to different components contribute $\frac{1}{\infty} = 0$ to the IGL. We can easily
compute the IGL in $O(n^2)$ time, and $O(n)$ space, on the real-RAM by traversing from each vertex. Hence, it remains to show that IGL is balanced on
trees: it suffices to show that there is a constant $t_{\text{IGL}}$, such that any
subset of vertices whose deletion minimizes the IGL is
$O\left(\left(n/k\right)^{t_{\text{IGL}}} \right)$-trimming.

To do so, we choose to reason about the decrease in IGL caused by the removal
of a subset of $k$ vertices, rather than the IGL itself. Maximizing this
decrease (which we call \emph{utility}) is equivalent to minimizing the IGL of
the graph after removal.

\begin{definition}[Utility]
  \label{def-utility}
  Let $G = (V, E)$ be a graph. Then the utility of some $S \subseteq V$ is:
  \[
    \mathcal{U}_G(S) = IGL(G) - IGL(G - S).
  \]
  If $S = \{v\}$% is a singleton, 
, we write $\mathcal{U}_G(v)$ instead of
  $\mathcal{U}_G(\{v\})$, which we call the utility of $v$ in $G$.
\end{definition}

\noindent
Suppose $S = S' \dot\cup \{v\}$ is a subset of $k$ vertices in a tree $T$
with maximum utility. Necessarily, $v$ must have maximum utility in $T - S'$.
This means that $v$ has no less utility than any vertex in its component in $T
- S'$, and that it also has no less utility than the optimal vertex in any
other component. In this vein, we would like to consider the case when $k = 1$
so we can reason about the individual optimality of each vertex in an optimal
solution.

By first considering the sum of vertex utilities, we prove the following upper and lower bounds on the utility of the optimal choice of vertex, in this case.

% When $k = 1$, the \textsc{MinIGL} problem seeks the vertex with
% maximum utility, and so we derive the following results concerning the utility
% of individual vertices for this case.

\begin{restatable}{lemma}{lemtotalutility}
  \label{total-utility}
  Let $T = (V, E)$ be a tree with $n$ vertices. Then
  \[
    \sum_{v \in V} \mathcal{U}_T(v) = \binom{n}{2} + IGL(T).
  \]
\end{restatable}

\begin{proof}
  Since $T$ is a tree, there is precisely one simple path with length $d_T(u,
  v)$ for every $\{u, v\} \subseteq V$. Now suppose $w \in V$ and compare
  $d_{T-w}(u, v)$ to $d_T(u, v)$. If $w \in \mathcal{P}_T(u, v)$, then $u$ and $v$
  are necessarily disconnected in $T-w$, so $d_{T-w}(u, v) = \infty$. Otherwise,
  the path remains, and $d_{T-w}(u, v) = d_T(u, v)$.

  Hence, the utility of $w$ in $T$ can be written as follows.
  \begingroup
  \allowdisplaybreaks
  \begin{align*}
    \mathcal{U}_T(w)
  &= IGL(T) - IGL(T - w) \\
  &= \sum_{\{u, v\} \subseteq V} \frac{1}{d_T(u, v)} - \frac{1}{d_{T-w}(u, v)} \\
  &= \sum_{\substack{\{u, v\} \subseteq V \\ w \in \mathcal{P}_T(u, v)}}
    \frac{1}{d_T(u, v)}. \\
    \intertext{We can then sum this utility over all vertices in $T$.}
    \sum_{w \in V} \mathcal{U}_T(w)
    &= \sum_{w \in V} \sum_{\substack{\{u, v\} \subseteq V \\ w \in \mathcal{P}_T(u, v)}}
    \frac{1}{d_T(u, v)} \\
    &= \sum_{\{u, v\} \subseteq V} \sum_{w \in \mathcal{P}_T(u, v)} \frac{1}{d_T(u, v)} \\
    &= \sum_{\{u, v\} \subseteq V} \abs{\mathcal{P}_T(u, v)} \frac{1}{d_T(u, v)} \\
  &= \sum_{\{u, v\} \subseteq V} \frac{d_T(u, v) + 1}{d_T(u, v)} \\
  &= \binom{n}{2} + IGL(T),
  \end{align*}
  \endgroup
  as required.
\end{proof}

\begin{restatable}{lemma}{lemmaxtreeutilitylb}
  \label{max-tree-utility-lb}
  Let $T = (V, E)$ be a tree with $n \geq 2$ vertices. Then, $\max_{v \in V}
  \mathcal{U}_T (v) \geq n/2$.
\end{restatable}

\begin{proof}
  \sloppy
  From \autoref{total-utility}, we know the average utility among the vertices in
  $T$ is $\frac{1}{n} \left( \binom{n}{2} + IGL(T) \right) = \frac{n-1}{2} +
  \frac{1}{n} IGL(T)$. Since $T$ has $n - 1$ edges, each connecting a pair of vertices
  distance 1 apart, $IGL(T) \geq n-1$ so $\frac{1}{n} IGL(T) \geq \frac{n-1}{n}
  \geq \frac{1}{2}$ whenever $n \geq 2$. Hence, the average utility is at least
  $n/2$ and the result follows.
\end{proof}

% If the diameter of the tree is bounded, we can improve the lower bound given in
% \autoref{max-tree-utility-lb} by considering the utility of a centroid.
% First, we give a lower bound on the number of paths through a centroid.
% 
% By considering the utility of the centroid, we obtain the following lower bound
% for the maximum utility of any vertex, in terms of the diameter of the tree.

% \begin{corollary}
%   \label{cor:max-tree-utility-lb-diam}
%   Let $T$ be a tree with $n \geq 2$ vertices and diameter no greater than some
%   constant $D$ and suppose $u$ is a centroid of $T$. Then, $\max_{v \in V}
%   \mathcal{U}_T (v) \geq \mathcal{U}_T(u) \geq \frac{n^2}{4D}$.
% \end{corollary}

\begin{restatable}{lemma}{lemmaxtreeutilityub}
  \label{max-tree-utility-ub}
  Let $G = (V, E)$ be a tree with $n$ vertices. Then $\mathcal{U}_G(v) \leq IGL(G)
  \leq \frac{1}{2} n(n-1)$ for any vertex $v \in V$.
\end{restatable}

\begin{proof}
  From Definition \ref{def-utility}, we have that $\mathcal{U}_G(v) = IGL(G) - IGL(G
  - v)$. It is obvious that $IGL(G - v) \geq 0$, so $\mathcal{U}_G(v) \leq IGL(G)$.
  The second inequality holds because there are $\binom{n}{2} = \frac{1}{2}
  n(n-1)$ vertex pairs in the graph, and each contributes at most 1 to the IGL.
  %\sg{Change this to: there are $\binom{n}{2} = \frac{1}{2}
  %	n(n-1)$ vertex pairs in the graph}\sg{This bound holds for graphs; for trees one can do better since at most $n-1$ pairs are at distance one.}
  %\sg{Move the discussion on equality after the proof, since it does not contribute to the proof.}
\end{proof}

\noindent
Note that equality holds for the second inequality precisely when $G$ is
the complete graph $K_n$, on $n$ vertices. On trees, one can show that the IGL is at most
$\frac{1}{4}(n^2 + n - 2)$, with equality when the tree is a star.
%We choose to use the slightly weaker result on general graphs, for ease of computation.
% \sg{This should be $n-1 + \frac{(n-1)(n-2)}{4}$.}
% \jl{I changed the result to be for a general graph, which is weaker, but is
% enough for our purposes.}

Next, we show that the removal of a vertex with maximum utility leaves the remaining forest somewhat balanced.
Specifically, it is never the case that one tree in this forest is so large that it contains all but $o(n^{1/4})$ vertices.

\begin{restatable}{theorem}{thmlargestcomp}
  \label{largest-component-min-size}
  Let $T = (V, E)$ be an unweighted tree with $n \geq 3$ vertices and suppose
  $v \in V$ minimizes $IGL(T - v)$.
  %, that is, $v$ is some vertex whose deletion minimizes the IGL of the resultant forest. %SG: I don't think this is necessary to rephrase
  Further, suppose $C$ is a
  connected component in $T - v$ containing $l$ vertices and let $r = n - l -
  1$ be the number of vertices in $T - v$ not in $C$. Then, there is a constant
  $0 < c < 1$ independent of $n$ such that $r \geq c n^{1/4}$.
  % for any tree with $n$ vertices. %SG: I removed this since the tree was already fixed before.
\end{restatable}

\begin{proof}
  %\jl{Explain why r >= 1?}
  %\sg{does not seem necessary to me.}
  % \jl{Ended up stating trivial cases, since the details seem a little relevant
  % later on.}
  We may assume $l \geq 1$, since the case when $l = 0$ is trivial. We may also
  assume that $r \geq 1$, since if $r = 0$, $v$ is a leaf, which contradicts
  its optimality since $n\ge 3$.
  
  Since $T$ is a tree, each neighbour of $v$ belongs to a different
  component in $T - v$. Suppose $x_C$ is the neighbour of $v$ in $C$ and let
  $C'$ be the subtree $T[V(C) \cup \{v\}]$. Thus,
  $v$ is a leaf of $C'$. We use this structure (pictured in
  Figure~\ref{fig:arg-pos}) to give two different, but related, upper bounds
  for the utility $\mathcal{U}_T(v)$ of $v$ in $T$.

  % \jl{Insert diagram.}
  \begin{figure}
    \tikzstyle{vertex}=[circle,draw=black,thick,minimum size=5pt,inner sep=0pt]
    \tikzstyle{selected vertex} = [vertex, fill=red!24]
    \tikzstyle{edge} = [draw,thick,-]
    \tikzstyle{weight} = [font=\small]

    \begin{center}
        \begin{tikzpicture}[scale=0.8]
        \usetikzlibrary{decorations.pathmorphing}

        % \draw (3.5, 5.9) rectangle (6, 8);
        % \draw (1, 0.6) rectangle (5, 3.5);
        % \draw (-0.8, -1.7) rectangle (0.8, -3);
        % \draw (1.2, -1.7) rectangle (2.8, -3);
        % \draw (3.0, -1.7) rectangle (5.0, -3);

        \draw node[] () at (0, 0.3) {$v$};
        \draw node[] () at (1, 0.3) {$x_C$};

        \draw node[vertex] (aa) at (-2, 0.9) {};
        \draw node[vertex] (ab) at (-2, 0.5) {};
        \draw node[vertex] (ca) at (-2, -0.7) {};
        \draw node[vertex] (a) at (-1, 0.7) {};
        \draw node[vertex] (b) at (-1, 0) {};
        \draw node[vertex] (c) at (-1, -0.7) {};
        \draw node[vertex] (v) at (0, 0) {};
        \draw node[vertex, fill=gray!30] (xc) at (1, 0) {};

        \draw node[vertex, fill=gray!30] (d) at (2, 0.7) {};
        \draw node[vertex, fill=gray!30] (e) at (2, -0.7) {};

        \draw node[vertex, fill=gray!30] (f) at (3, 1.0) {};
        \draw node[vertex, fill=gray] (g) at (3, 0.4) {};
        \draw node[] () at (3, 0.1) {$v_A$};

        \draw node[vertex, fill=gray!30] (h) at (4, 0.8) {};
        \draw node[vertex, fill=gray!30] (i) at (4, 0.4) {};
        \draw node[vertex, fill=gray!30] (j) at (4, 0.0) {};
        \draw node[vertex, fill=gray!30] (k) at (4, -0.4) {};

        \draw node[vertex, fill=gray!30] (l) at (5, 1.0) {};
        \draw node[vertex, fill=gray!30] (m) at (5, 0.6) {};
        \draw node[vertex, fill=gray!30] (n) at (5, -0.2) {};
        \draw node[vertex, fill=gray!30] (o) at (5, -0.6) {};

        \draw node[vertex] (p) at (6, 1.0) {};
        \draw node[vertex] (q) at (6, -0.6) {};

        \draw node[vertex] (qa) at (7, -0.2) {};
        \draw node[vertex] (qb) at (7, -1.0) {};
        \draw node[vertex] (qc) at (7, -0.6) {};

        \foreach \source / \dest in
        {a/v, b/v, c/v, v/xc, xc/d, xc/e, d/f, d/g, g/h, g/i, g/j, g/k,
         h/l, h/m, k/n, k/o, l/p, o/q, aa/a, ab/a, ca/c, q/qa, q/qb, q/qc}
        \path[edge] (\source) -- node[weight] {} (\dest);

        \draw
          [decorate,decoration={brace,amplitude=10pt},xshift=-4pt,yshift=0pt]
          (5.3,-1.2) -- (0,-1.2) node [black,midway,xshift=-0.6cm] {};
        \draw node[] () at (2.5, -1.8) {\small $T_A$};

        \draw
          [decorate,decoration={brace,amplitude=10pt},xshift=-4pt,yshift=0pt]
          (7.6,-1.2) -- (6,-1.2) node [black,midway,xshift=-0.6cm] {};
        \draw node[] () at (6.5, -1.8) {\small vertices of $B$};

        \draw
          [decorate,decoration={brace,amplitude=10pt},xshift=-4pt,yshift=0pt]
          (7.3,-2) -- (0,-2) node [black,midway,xshift=-0.6cm] {};
        \draw node[] () at (3.5, -2.7) {\small $C'$};

        \draw
          [decorate,decoration={brace,amplitude=10pt},xshift=-4pt,yshift=0pt]
          (1,1.5) -- (7.3,1.5) node [black,midway,xshift=-0.6cm] {};
        \draw node[] () at (4, 2.2) {\small $C$, containing $l$ vertices};

        \draw
          [decorate,decoration={brace,amplitude=10pt},xshift=-4pt,yshift=0pt]
          (-2.2,1.5) -- (-0.6,1.5) node [black,midway,xshift=-0.6cm] {};
        \draw node[align=center] () at (-1.5, 2.2) {\small $r$ vertices};
      \end{tikzpicture}

      \caption{Layout of the vertices of $T$, in
        Theorem~\ref{largest-component-min-size}.  Shaded vertices are in $A$,
        and are no more than $D = 5$ away from $v$.  The value of $D$ here has
        chosen for example's sake, and is not the true value constructed in the
      proof.}
    \label{fig:arg-pos}
    \end{center}
  \end{figure}
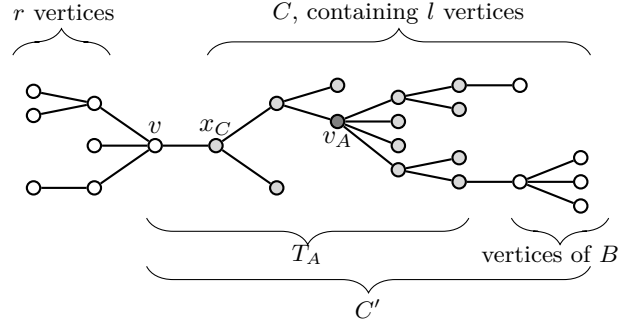

  \begin{restatable}{claim}{claimvutilityub}
    \label{v-utility-ub}
    $\mathcal{U}_T(v) \leq \frac{1}{2} r(r+1) + (r+1) \mathcal{U}_{C'}(v)$.
  \end{restatable}

  \begin{proof}
      Let us upper bound $\mathcal{U}_T(v)$ by considering the utility of $v$ in $C'$
      and also in $T - V(C)$. There are $n - l$ vertices in $T - V(C)$, so by
      Lemma \ref{max-tree-utility-ub}, we have that $\mathcal{U}_{T - V(C)}(v) \leq
      \frac{1}{2} (n-l)(n-l-1) = \frac{1}{2} r(r+1)$. This accounts for the pairs
      of vertices disconnected by the deletion of $v$ in $T - V(C)$.

      We still need to consider such pairs where one vertex is in $C$, and the
      other is in $T - V(C)$ (this includes $v$). Since $v$ is a leaf in
      $C'$, the only pairs of vertices connected in $C'$ that are disconnected in
      $C = C' - v$ are those of the form $\{v, v_C\}$, where $v_C$ ranges over
      $V(C)$. Now let $u$ be a vertex in $V \setminus V(C)$. The path from $u$ to
      $v_C$ must pass through $v$, and thus $d_T(u, v_C) \geq d_T(v, v_C)$. Hence,
      the contribution of each disconnected $\{u, v_C\}$ pair is at most that of
      $\{v, v_C\}$ towards $\mathcal{U}_T(v)$. Putting these inequalities together
      gives us
      \begin{align*}
          \mathcal{U}_T(v)
          &= \sum_{\{p, q\} \in \mathcal{P}^{-1}_T(v)} \frac{1}{d_T(p, q)} \\
          &= \mathcal{U}_{T - V(C)}(v) +
          \sum_{\substack{u \in V \setminus V(C) \\ v_C \in V(C)}} \frac{1}{d_T(u, v_C)} \\
          &\leq \frac{1}{2} r(r+1) + 
          \abs{V \setminus V(C)} \sum_{v_C \in V(C)} \frac{1}{d_T(v, v_C)} \\
          &\leq \frac{1}{2} r(r+1) + (r+1) \sum_{v_C \in V(C)} \frac{1}{d_T(v, v_C)} \\
          &= \frac{1}{2} r(r+1) + (r+1) \mathcal{U}_{C'}(v),
      \end{align*}
      as required.
  \end{proof}

  \begin{restatable}{claim}{claimlooseleafutility}
    \label{loose-leaf-utility}
    $\mathcal{U}_T(v) \leq rn$.
  \end{restatable}
  
  \begin{proof}
      Since $v$ is a leaf of $C'$, it is distance 1 away from its sole neighbour,
      and only this neighbour, in $C'$. Also, the only pairs disconnected by
      $v$'s removal in $C'$ are those containing $v$ itself. Now there are $l-1$
      other vertices in $C'$, each at least distance 2 away from $v$.  Hence,
      $\mathcal{U}_{C'}(v) \leq 1 + \frac{l-1}{2} = \frac{l+1}{2} = \frac{n-r}{2}$.

      Since $r \geq 1$, we know that $r + 1 \leq 2r$. Hence, by Claim \ref{v-utility-ub} 
      \begin{align*}
          \mathcal{U}_T(v) &\leq r^2 + 2r \mathcal{U}_{C'}(v) \\
                           &\leq r^2 + r(n-r) \\
                           &= rn,
      \end{align*}
      as required.
  \end{proof}

  \noindent
  Since the utility of deleting $v$ is maximal among all vertices, and $n \geq
  2$, we know $\mathcal{U}_T(v) \geq n/2$ from \autoref{max-tree-utility-lb}.
  Combining this with Claim \ref{v-utility-ub} and rearranging gives
  \begin{align}
    \mathcal{U}_{C'}(v) &\geq \frac{n - r(r+1)}{2(r+1)}. \label{path-utility-lb}
  \end{align}
  Suppose, for a contradiction, that $r < \frac{1}{15}n^{1/4}$. Since $r$ is
  purported to be relatively small, $\mathcal{U}_{C'}(v)$ must be rather large (note
  it is proportional to $n$). Intuitively, this implies that many vertices in
  $C'$ are close to $v$, and hints towards a more central choice of vertex to
  delete. We will formally show that such a vertex exists, and is a more optimal
  choice.

  Fix some distance $D$. We can divide the vertices of $C$ into two groups, $A$
  and $B$: those at most distance $D$ from $v$ in $C'$ (and thus, also in $T$) and
  those that are not, respectively. Suppose that $\abs{A} = t$ and that $\abs{B}
  = \abs{V(C)} - t$. We have the following upper bound:
  \begin{align}
    \mathcal{U}_{C'}(v) &\leq t + \frac{\abs{V(C)} - t}{D + 1} \leq t + \frac{n - t}{D + 1}, \label{t-distribution}
  \end{align}
  because each vertex in $B$ is at least distance $D+1$ away from $v$, and
  $\abs{V(C)} \leq n$. Note that we do not account for $v$ itself, since
  the distance to itself does not contribute to its utility.

  Recall that $r < \frac{1}{15} n^{1/4}$. It is easy to see that $r(r+1) \leq
  n/2$. Combining this with
  \eqref{path-utility-lb} and \eqref{t-distribution} gives us the following
  inequality:
  \begin{align*}
    \frac{n}{4(r+1)} \leq \mathcal{U}_{C'}(v) \leq t + \frac{n-t}{D+1},
  \end{align*}
  from which we can obtain
  \begin{align*}
    tD \geq \frac{n(D+1)}{4(r+1)} - n.
  \end{align*}
  If we choose $D = 8(r+1)-1$, it holds that $t \geq \frac{n}{D} =
  \frac{n}{8r+7} \geq \frac{n}{15r}$.

  Consider the subgraph (a tree) $T_A$ induced by the vertex set $A \dot\cup
  \{v\}$. $T_A$ contains at least two vertices as $v$ and $x_C$ both must be in
  $A$. Also, since $T_A$ is a tree, by Lemma~\ref{lem:centroid-half} it must
  have a centroid. Let one of the centroids of $T_A$ be $v_A$. The diameter of
  $T_A$ is at most $2D$,
  since every vertex in $T_A$ is at most distance $D$ from $v$.
  %Then, by Corollary~\ref{cor:max-tree-utility-lb-diam}, we have
  Combining this with \autoref{centroid-disconnect}, we have
  \begin{align*}
      \mathcal{U}_{T_A}(v_A) &\geq \frac{t^2}{8D} \geq \frac{n^2}{8D^3} \geq \frac{n^2}{8(15)^3r^3}.
  \end{align*}

  % \jl{Do we need a basic lemma that subgraph only decreases IGL, or is this
  % obvious? We should at least explain that we disconnect more pairs, since this
  % property depends on the tree structure.}
  % \sg{No need for a lemma, but we can write one sentence about this.}
  \noindent
  Now every pair in $T_A$ that is disconnected by the deletion of $v_A$ is also
  disconnected in $T$ by the deletion of $v_A$, so $\mathcal{U}_{T_A}(v_A) \leq
  \mathcal{U}_T(v_A)$. Also, by the optimality of $v$ in $T$, we have that
  $\mathcal{U}_T(v_A) \leq \mathcal{U}_T(v)$. Hence, using the result of Claim
  \ref{loose-leaf-utility}, we can conclude that
  \begin{align*}
    \frac{n^2}{8(15)^3r^3} \leq \mathcal{U}_{T_A}(v_A) \leq \mathcal{U}_T(v) \leq rn.
  \end{align*}
  Thus, we have that $r^4 \geq \frac{n}{8(15)^3} \geq \frac{n}{15^4}$,
  so $r \geq \frac{1}{15}n^{1/4}$, which is a contradiction. The result follows
  with a choice of $c = \frac{1}{15}$.
\end{proof}

We can use this result to finally upper bound the number of vertices in any
component after an optimal set of vertices has been removed.

\begin{theorem}
  \label{thm:largest-component-max-size}

  Let $T = (V, E)$ be a tree with $n$ vertices, and let $S \subseteq V$ be some
  subset of vertices such that $\abs{S} = k \geq 1$. There exists a positive
  constant $c'$, independent of $T$ and $k$, such that whenever $S$ minimizes
  $IGL(T - S)$, $S$ is $\left(c'\left( n/k \right)^5\right)$-trimming.
  
\end{theorem}

\begin{proof}

  We will call the components of $T - S$ \emph{remaining components} and denote
  each of them by their vertex set. Suppose the remaining
  components are $R = \{R_1, R_2, \dots, R_{\abs{R}}\}$, where $R_i \subseteq
  V$ and $S \dot\cup R_1 \dot\cup \dots \dot\cup R_{\abs{R}} = V$. We need to
  show that $\abs{R_i} \leq c'\left( n/k \right)^5$ for each $R_i$. 

  We first construct a new graph $T' = (V', E')$ by collapsing each of the remaining components. Formally, $V' = R
  \dot\cup S$, and, for each $R_i \in R$ and $s \in S$, $\{R_i, s\} \in E'$ if and only if there exists some $r \in R_i$
  such that $\{r, s\} \in E$. It can be seen that $T'$ is necessarily a tree, and that every $R_i$ is only incident to
  elements in $S$.  For the remainder of the proof, we further assume that every element in $S$ is only incident to
  remaining components in $T'$: if this is not the case, one can add a ``dummy'' remaining component with cardinality zero
  between every pair of adjacent elements of $S$ in $T'$.

  Let $R_r$ be a remaining component containing at least as many vertices as any other remaining component.
  Note that $\abs{R_r} > 0$: it is never a ``dummy''.
  It suffices to show the upper bound holds for $R_r$. We root $T'$ at
  $R_r$. Since $k > 0$, there are strictly fewer than $n$ vertices among the
  remaining components $R$. Hence, by the Pigeonhole Principle, there must be
  some $s \in S$ such that the children $R_{s_1}, R_{s_2}, \dots,
  R_{s_{ch(s)}}$ of $s$ in $T'$ together contain fewer than $n/k$
  vertices. Let the parent of $s$ in $T'$ be $R_{p(s)}$. %This is depicted in
  See Figure~\ref{fig:bounding-php}.

  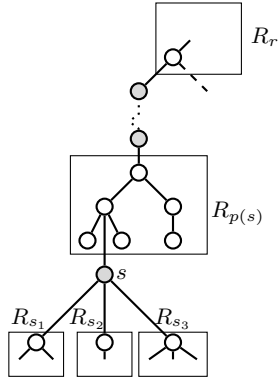
\begin{figure}
    \tikzstyle{vertex}=[circle,draw=black,thick,minimum size=6pt,inner sep=0pt]
    \tikzstyle{selected vertex} = [vertex, fill=red!24]
    \tikzstyle{edge} = [draw,thick,-]
    \tikzstyle{weight} = [font=\small]

    \begin{center}
      \begin{tikzpicture}[scale=0.45]
        \draw (3.5, 5.9) rectangle (6, 8);
        \draw (1, 0.6) rectangle (5, 3.5);
        \draw (-0.8, -1.7) rectangle (0.8, -3);
        \draw (1.2, -1.7) rectangle (2.8, -3);
        \draw (3.0, -1.7) rectangle (5.0, -3);

        \draw node[vertex] (r) at (4, 6.4) {};
        \draw node[vertex, fill=gray!30] (q) at (3, 5.4) {};
        \draw node[vertex, fill=gray!30] (p) at (3, 4) {};

        \draw node[vertex] (a) at (3, 3) {};
        \draw node[vertex] (b) at (2, 2) {};
        \draw node[vertex] (c) at (4, 2) {};
        \draw node[vertex] (d) at (1.5, 1) {};
        \draw node[vertex] (e) at (2.5, 1) {};
        \draw node[vertex] (f) at (4, 1) {};

        \draw node[vertex, fill=gray!30] (g) at (2, 0) {};
        %\draw node[vertex, fill=gray!30] (h) at (2, -1.7) {};
        \draw node[vertex] (i) at (0, -2) {};
        \draw node[vertex] (j) at (2, -2) {};
        \draw node[vertex] (k) at (4, -2) {};

        \foreach \source/ \dest in
        {a/b, a/c, b/d, b/e, c/f, b/g, g/i, g/j, g/k, p/a, r/q}
        \path[edge] (\source) -- node[weight] {} (\dest);
        \draw[edge, dotted, decorate, decoration={snake, amplitude=3pt, post length=0pt}] (p) -- node[weight] {} (q);
        \draw[edge] (i) -- node[weight] {} (-0.5, -2.5);
        \draw[edge] (i) -- node[weight] {} (0.5, -2.5);

        \draw[edge] (j) -- node[weight] {} (2, -2.5);

        \draw[edge] (k) -- node[weight] {} (3.3, -2.5);
        \draw[edge] (k) -- node[weight] {} (4, -2.5);
        \draw[edge] (k) -- node[weight] {} (4.7, -2.5);

        \draw[edge] (r) -- node[weight] {} (4.5, 6.9);
        \draw[edge, dashed] (r) -- node[weight] {} (5, 5.4);

        \node () at (6.7,7) {\footnotesize $R_r$};
        \node () at (5.9, 1.8) {\footnotesize $R_{p(s)}$};
        \node () at (2.5, 0) {\footnotesize $s$};
        \node () at (-0.2, -1.3) {\footnotesize $R_{s_1}$};
        \node () at (1.5, -1.3) {\footnotesize $R_{s_2}$};
        \node () at (4, -1.3) {\footnotesize $R_{s_3}$};

      \end{tikzpicture}

      \caption[Bounding the size of the largest remaining component]{Bounding
      the size of the largest remaining component in Theorem
    \ref{thm:largest-component-max-size}. In this case, the parent of $s$ is
  not the root and $s$ has $ch(s) = 3$ children in $T'$. The shaded vertices are
those in $S$.}
    \label{fig:bounding-php}
    \end{center}
  \end{figure}

  Since $S$ is optimal, $s$ must be an optimal choice of vertex to delete in an
  instance of \textsc{MinIGL} with graph $T - (S \setminus \{s\})$ and a budget of $1$
  deletion. In particular, it must also be the optimal choice of vertex to
  delete in the component containing $s$ in $T - (S \setminus \{s\})$. Hence, we
  may apply \autoref{largest-component-min-size} to $T - (S \setminus \{s\})$,
  in that component to give
  \begin{align*}
    \frac{n}{k} > \sum_{i=1}^{ch(s)} \abs{R_{s_i}}
    &\geq c (\abs{R_{p(s)}} + \abs{R_{s_1}} + \dots +
    \abs{R_{s_{ch(s)}}})^{1/4} \geq c \abs{R_{p(s)}}^{1/4},
  \end{align*}
  since $c > 0$, where $c$ is the constant in \autoref{largest-component-min-size}.  Thus, we have $\abs{R_{p(s)}} \leq c^{-4} \left( n/k \right)^4$.

  We now have two cases: if the parent $R_{p(s)}$ of $s$ in $T'$ is the root,
  $R_r$, or if it is not the root. If $R_{p(s)}$ is the root, then $p(s) = r$, so
  $\abs{R_{r}} \leq c^{-4} \left( n/k \right)^4$.
  Otherwise, $s$ is not a child of the root, and so $s$ must have been a more
  optimal choice than the best choice in the component induced by $R_r$ in $T -
  (S \setminus \{s\})$. Since this component contains $\abs{R_r}$ vertices, the
  best choice had utility at least $\frac{\abs{R_r}}{2}$, by
  Corollary~\ref{max-tree-utility-lb}. Now the paths that pass through $s$ in
  $T - (S \setminus \{s\})$ must have one endpoint in some $R_{s_j}$ and the other
  either in another $R_{s_j'}$ or in $R_{p(s)}$. This is the case since no path
  can have both endpoints in $R_{p(s)}$. Hence, there are at most $(n/k
  + 1)(n/k + \abs{R_{p(s)}})$ such pairs, accounting also for those paths
  starting at $s$. Since each of these paths have length at least 1, we have
  that
  \begin{align*}
    \frac{\abs{R_r}}{2} \leq \mathcal{U}_{T - (S \setminus \{s\})}(s)
    &\leq \left( \frac{n}{k} + 1 \right) \left( \frac{n}{k} + \abs{R_{p(s)}} \right) \leq 2\frac{n}{k} \left( \frac{n}{k} + \abs{R_{p(s)}} \right), 
    \intertext{because $k \leq n$. Thus}
    \frac{\abs{R_r}}{2}&\leq 2\frac{n}{k} \left( \frac{n}{k} + c^{-4} \left( \frac{n}{k} \right)^4  \right) 
    \leq 4c^{-4} \left( \frac{n}{k} \right)^5,
  \end{align*}
  because $0 < c < 1$. Hence, $\abs{R_r} \leq 8c^{-4} \left( n/k \right)^5$ and the result follows with a choice of $c'
  = 8c^{-4}$.
\end{proof}
\noindent
Thus, we can choose $c_{\text{IGL}} = c'$ and $t_{\text{IGL}} = 5$, showing that IGL is indeed balanced on trees. 
This gives \autoref{thm:fptstatement} and \autoref{thm:subexp-statement}.

\section{Computing the IGL}

\label{sec:compute-igl}

% \jl{Would this be cleaner if we extended notation to handle IGL of a set of
% pairs?}
% \jl{Citation needed for FFT: do we need to do this over a finite field? What
% about precision? Talk to David Harvey}

Computing the IGL of a graph is trivial once its distance distribution has been
determined. In this section, we describe algorithms for efficiently computing
the distance distribution of trees, and extend these ideas to graphs with small treewidth.

%\jl{Update this section}

% \sg{I suggest to move this section into the appendix.}
% \jl{Cool!}

\subsection{Trees}

To compute the distance distribution on trees, we present a divide-and-conquer method (commonly known as the
\emph{centroid decomposition}, as used in \cite{Bilo2015}) as follows. We pick a vertex and compute the contribution to
the distance distribution of all paths passing through that vertex, using fast polynomial multiplication. Then, we
delete the vertex from the tree, and recurse on the remaining connected subtrees. We first provide a method that
efficiently computes this contribution.

\begin{lemma}
  \label{thm:compute-root-igl}
  Let $T = (V, E)$ be an unweighted tree with $n$ vertices and suppose $r \in
  V$. Then, the contribution to the distance distribution
  of all pairs in $\mathcal{P}^{-1}_T(r)$ can be found in $O(n \log n)$ time on a log-RAM.
\end{lemma}

\begin{proof}

  We begin by rooting the tree at $r$.  Suppose the children of $r$ are $s_1,
  \dots, s_{ch(r)}$ and let $S_1, \dots, S_{ch(r)}$ denote the set of vertices
  in the subtrees rooted at each child, respectively. With the addition of $S_0
  = \{r\}$, the sets $S_i$ form a partition of $V$.
  
  We perform a depth-first search from $r$, to find $d_T(r, u) = d_T(u, r)$ for
  each vertex $u$ and construct a sequence of distance polynomials $P_0, P_1,
  \dots, P_{ch(r)}$, where $P_i(x) = \sum_{w \in S_i} x^{d_T(r, w)}$. This
  takes $O(n)$ time, storing each distance polynomial in coefficient form: there are at most $n$ terms overall.
  Now let
  
  \[
  P(x) = \left( \sum_{0 \leq i \leq ch(r)} P_i(x) \right)^2 -
    \left( \sum_{0 \leq i \leq ch(r)} P_i^2(x) \right) 
      = \sum_{0 \leq j \leq n} b_j x^j.
  \]

  \noindent
  We observe that
  \begin{align}
    b_j = 2\abs{ \{ \{u, v\} \in \mathcal{P}^{-1}_T(r) : u \neq v \text{
    and } d_T(u, v) = j \} }, \label{square-sum-tree-igl}
  \end{align}
  that is, $b_j$ is twice the number of pairs of distinct vertices which have
  a path of length $j$ passing through $r$.
  Thus, the required contribution to the distance distribution can be read off from the coefficient form of $P(x)$.
  %, so
  %it remains to show that this can be computed efficiently from the coefficients of each $P_i$.
  The result follows by computing this efficiently from the coefficients of each $P_i$ by applying
  \autoref{lem:fast-mult}, and observing
  that the degree of $\sum_{0 \leq i \leq ch(r)} P_i(x)$, and the sum of the degrees of the $P_i$'s are
  both at most $n$.
\end{proof}

If we always pick $r$ in \autoref{thm:compute-root-igl} to be a centroid of the
tree, Lemma~\ref{lem:centroid-half} ensures that each vertex can appear in at
most $\log_2 n + 1$ trees throughout the execution of our divide-and-conquer
algorithm. A centroid must always exist (also by
Lemma~\ref{lem:centroid-half}), and we can find one in linear time by
recursively computing, then examining, subtree sizes. This gives
Theorem~\ref{thm:dist-trees}.

\thmdisttrees*

If we only wish to determine the first $p$ values of the distance distribution
of $T$, we can modify \autoref{thm:compute-root-igl} to run in $O(n + p
\log n)$ time, by discarding all terms with degree greater than $p$ when constructing the
polynomials. Thus, the expensive multiplication step costs $O(p \log n)$ time
by Lemma~\ref{lem:fast-mult}, and we obtain \autoref{thm:tree-prefix} as a corollary.

\begin{restatable}{theorem}{thmtreeprefix}
  \label{thm:tree-prefix}
  The prefix $a_1, \dots, a_p$ of the distance distribution of a tree with $n$
  vertices can be computed in $O(n \log n + p \log^2 n)$ time on a log-RAM.
\end{restatable}

\subsection{Graphs with small treewidth}

Here, we extend the ideas used in the previous section to prove Theorem~\ref{thm:dist-pref}.

Let $G = (V, E)$ be
an undirected graph with $n$ vertices, whose edges each have a non-negative
weight.
We describe a modification of the method of Cabello and Knauer \cite{Cabello2009}, to recursively reduce the task of computing
the distance distribution of $G$ to solving instances of \textsc{RedBluePolynomial} over points in $O(\text{tw}(G))$
dimensions, with values at most $p$.
%\jl{Double-check this}
%The method for this reduction is detailed in Appendix~\ref{subs:reduction}.

%Let $G = (V, E)$ be an undirected graph with $n$ vertices, whose edges each have a non-negative weight.
In time $2^{O(k)} n$, we can compute a tree decomposition of $G$ of
width at most $k = 5 \cdot \text{tw}(G) + 4$ containing at most $O(kn)$ nodes
\cite{Bodlaender2016}. Using a common technique, we can transform this
decomposition into a \emph{nice} tree decomposition with $N = O(kn)$ nodes
(see, for example \cite{Cygan2015}).
The nodes of a nice tree decomposition form a rooted binary tree.
%Among other properties, a nice tree
%decomposition is a binary tree. %SG: this is also stated below

Let $A$ be a subset of vertices. A \emph{portal} of $A$ is a vertex in $A$ which has, as a
neighbour, some vertex outside $A$. If these portals are contained in some set
$S \subseteq A$, we can partition the vertices of the graph into three sets: $A \setminus S, S$ and $V
\setminus A$, such that every path from a vertex in $A$ to a vertex in $V
\setminus A$ passes through some vertex in $S$. 

% \jl{It may be simpler just to cite the Cabello and Knauer result, if we can
%   confirm how it works. In any case it is not necessary to obtain the required
% time complexity.}
%\sg{Can you remind me if using their result would improve the time complexity?}
%\jl{Not anymore, since we have $O(k)$ in the exponent. I think it is okay to
%keep our own explanation here.}

Since the nice tree decomposition is a binary tree, there is some edge
$ij$ in the decomposition whose removal splits the
decomposition's tree into two components $I$ and $J$ (containing nodes $i$ and $j$,
respectively), each containing between $\frac{N}{3}$ and $\frac{2N}{3}$
nodes. Let $A$ be the set of vertices that appear in component $I$, and let $S$
be the intersection $B_i \cap B_j$ of the bags corresponding to nodes $i$ and
$j$.  Necessarily, $S$ must contain all the portals of $A$ due to properties of
the tree decomposition.  Moreover, $\abs{B_i \cap B_j} \leq \min( \abs{B_i},
\abs{B_j}) \leq k+1$.

Given this fixed $A$, recursively find the distance distribution among all
pairs of vertices in $A$ as follows. First, perform Dijkstra's algorithm from
all vertices in $S$.  For every pair of vertices in $S$, add an edge whose
weight equals the length of the shortest path between them. After these edges
are added, the length of the shortest path in $G$ between any pair of vertices
in $A$ can be found by only considering paths passing through the vertices of
$A$. Hence, we remove all vertices in $V \setminus A$, and recurse on this
smaller graph. Note that $I$ is a valid tree decomposition for this new graph,
since $i \in I$, and all the added edges have their endpoints in $B_i$. Thus,
we do not need to find another tree decomposition for this new graph, and its
treewidth does not exceed $k$.

In the same way, we recursively find the distance distribution induced by the
pairs of vertices in $(V \setminus A) \dot\cup S$. Between these two sets of
pairs, we have counted pairs of vertices in $S$ twice, so we subtract the
distance distribution induced by these pairs using the shortest paths already
computed.

Finally, we must compute the distance distribution among shortest paths between
the remaining pairs of vertices: these are the pairs in $(A \setminus S) \times
(V \setminus A)$. Let the vertices in $S$ be $s_1, \dots, s_{\abs{S}}$. For
every $a$ in $A \setminus S$ and every $b$ in $V \setminus A$, we associate
$(a, b)$ with precisely one vertex in $s$ through which some shortest path
between the vertices passes. More formally, we will associate $(a, b)$ with the
only $s_i$ such that
\begin{align*}
  d_G(a, s_i) + d_G(s_i, b) &< d_G(a, s_j) + d_G(s_j, b) && \text{for all $j < i$, and} \\
  d_G(a, s_i) + d_G(s_i, b) &\leq d_G(a, s_j) + d_G(s_j, b) && \text{for all $j > i$.}
  \intertext{By rearranging, and observing that all distances are integers, we deduce that this is precisely when}
  d_G(a, s_i) - d_G(a, s_j) &< d_G(s_j, b) - d_G(s_i, b) && \text{for all $j < i$, and} \\
  d_G(a, s_i) - d_G(a, s_j) &< d_G(s_j, b) - d_G(s_i, b) + 1 && \text{for all $j > i$.}
\end{align*}
\noindent
Note that all these distances are known from our application of Dijkstra's
algorithm from each vertex in $S$. Since any path from $a$ to $b$ must pass
through $S$, it follows from these inequalities that $d_G(a, s_i) + d_G(s_i, b)
= d_G(a, b)$. For each vertex $s_i$ in turn, we will compute the contribution
of all pairs of vertices associated with $s_i$ to the distance distribution. We
do so by reducing this task to an instance of \textsc{RedBluePolynomial}.

Our instance will have points in $\abs{S}-1$ dimensions: one dimension for each $j \neq i$. For each $a \in A
\setminus S$, create a red point with coordinate $d_G(a, s_i) - d_G(a, s_j)$
in the dimension corresponding to $j$, and value $d_G(a, s_i)$, corresponding to each $s_j \neq
s_i$. Similarly, for each $b \in V \setminus A$, create a blue point with
coordinate $d_G(s_j, b) - d_G(s_i, b)$, for each $j < i$, and $d_G(s_j, b) -
d_G(s_i, b) + 1$ for each $j > i$, with value $d_G(s_i, b)$. Importantly, we
omit any points with value greater than $p$: these cannot contribute to the
prefix we are trying to compute. Hence, we have created no more than $n$ points
in all, each with a non-negative integer value no greater than $p$. The
coefficient of $x^l$ produced by our instance of \textsc{RedBluePolynomial}
corresponds to the number of pairs associated with $s_i$ that are distance $l$
apart.
This concludes the description of our reduction.

Naturally, we now turn our attention to solving \textsc{RedBluePolynomial}.
Naively, this can be done in quadratic time by considering every pair of
points. However, when values are bounded -- such as in our instance -- we can
solve the problem more efficiently. 

For a given instance of \textsc{RedBluePolynomial}, let $n = r + b$ be the
total number of points and suppose the value of each point does not exceed some
integer $v \geq 0$. Below, we give solutions with time complexity parameterized
by both $n$ and $v$. We consider the 1-dimensional case, then
extend this result to higher dimensions recursively.

\begin{restatable}{lemma}{lemsqrt}
  \label{bw-one-d}

  \sloppy
  When $d = 1$, there is an algorithm that solves \textsc{RedBluePolynomial}
  in $O(n \sqrt{v \log n} + n \log n)$ time on a log-RAM.
\end{restatable}

\begin{proof}

  Sort the red and blue points together in non-decreasing order of the coordinate, placing blue points earlier in the order when there are
  ties. Let $t$ be a positive integer no greater than $n$. Assign points to
  groups of size no more than $t$ by placing the first $t$ points, in order,
  into a group $G_1$, followed by the next $t$ points in order into a group
  $G_2$, and so on, so we create $\ceil{\frac{n}{t}}$ groups in all. An example
  is given in Figure~\ref{fig:sqrt-decomp}.

  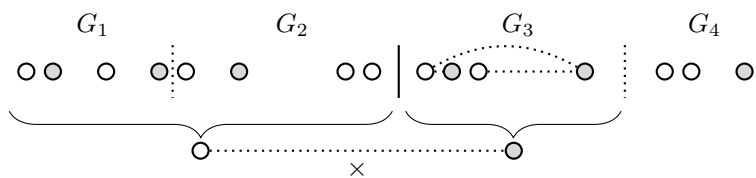
\begin{figure}
    \tikzstyle{vertex}=[circle,draw=black,thick,minimum size=6pt,inner sep=0pt]
    \tikzstyle{selected vertex} = [vertex, fill=red!24]
    \tikzstyle{edge} = [draw,thick,-]
    \tikzstyle{weight} = [font=\small]

    \begin{center}
      \begin{tikzpicture}[node distance=1cm, scale=0.35]
        \usetikzlibrary{positioning}

        \draw node[vertex] (1) at (-10, 0) {};
        \draw node[vertex, fill=gray!30] (2) at (-8, 0) {};
        \draw node[vertex] (3) at (-7, 0) {};
        \draw node[vertex, fill=gray!30] (4) at (-5, 0) {};
        \draw node[vertex] (5) at (-1, 0) {};
        \draw node[vertex] (6) at (0, 0) {};
        \draw node[vertex] (7) at (2, 0) {};
        \draw node[vertex, fill=gray!30] (8) at (3, 0) {};
        \draw node[vertex] (9) at (4, 0) {};
        \draw node[vertex, fill=gray!30] (10) at (8, 0) {};
        \draw node[vertex] (11) at (11, 0) {};
        \draw node[vertex] (12) at (12, 0) {};
        \draw node[vertex, fill=gray!30] (13) at (14, 0) {};
        \draw node[vertex, fill=gray!30] (14) at (-12, 0) {};
        \draw node[vertex] (15) at (-13, 0) {};

        \draw[edge, dotted] (-7.5, 1) -- node[] {} (-7.5, -1);
        \draw[edge] (1, 1) -- node[] {} (1, -1);
        \draw[edge, dotted] (9.5, 1) -- node[] {} (9.5, -1);

        \draw[edge, dotted] (9) -- node[] {} (10);
        \draw[edge, dotted, bend left=30] (7) edge node[] {} (10);
        \draw[edge, dotted] (7) edge node[] {} (8);

        % \draw[edge] (i) -- node[weight] {} (-0.5, -2.5);
        % \draw[edge] (i) -- node[weight] {} (0.5, -2.5);

        % \draw[edge] (j) -- node[weight] {} (2, -2.5);

        % \draw[edge] (k) -- node[weight] {} (3.3, -2.5);
        % \draw[edge] (k) -- node[weight] {} (4, -2.5);
        % \draw[edge] (k) -- node[weight] {} (4.7, -2.5);

        % \draw[edge] (r) -- node[weight] {} (4.5, 6.9);
        % \draw[edge, dashed] (r) -- node[weight] {} (5, 5.4);
        \draw
          [decorate,decoration={brace,amplitude=10pt},xshift=-4pt,yshift=0pt]
          (0.9,-1.5) -- (-13.5,-1.5) node [black,midway,xshift=-0.6cm] {};

        \draw
          [decorate,decoration={brace,amplitude=10pt},xshift=-4pt,yshift=0pt]
          (9.5, -1.5) -- (1.4,-1.5) node [black,midway,xshift=-0.6cm] {};

        \draw node[vertex] (rw) at (-6.45, -3) {};
        \draw node[vertex, fill=gray!30] (rb) at (5.325, -3) {};
        \draw[edge, dotted] (rw) -- node[below] {$\times$} (rb);

        \draw node[] () at (-10.5, 1.75) {$G_1$};
        \draw node[] () at (-3, 1.75) {$G_2$};
        \draw node[] () at (5.5, 1.75) {$G_3$};
        \draw node[] () at (12.5, 1.75) {$G_4$};
      \end{tikzpicture}
      \caption{Square root decomposition in Lemma~\ref{bw-one-d}. The empty
        dots represent red points, and the shaded dots represent blue points.
        When processing $G_3$, we consider pairs of points within the group
        where the red point precedes the blue point. We then consider
      cross-group pairs whose blue point is in $G_3$ using fast polynomial
    multiplication.}
    \label{fig:sqrt-decomp}
    \end{center}
  \end{figure}

  We will separately consider pairs of points that both belong to the same
  group, and those that belong to different groups. In each group, consider
  every pair of points, and check if they contribute a term to the polynomial.
  This takes $O(nt)$ time over all groups.
  
  It remains to consider pairs that belong to different groups: call these
  \emph{cross-group} pairs. For each blue point in $G_i$, we must add an extra
  term for each red point among $G_1, \dots, G_{i-1}$.  Thus, the total
  cross-group contribution of all pairs with a blue point in $G_i$ can be
  written as the following product of two polynomials.
  \[
    \sum_{B_q \in G_i} \; \sum_{R_p \in G_1 \cup \dots \cup G_{i-1}} x^{r_p + b_q}
      =
      \left(\sum_{R_p \in G_1 \cup \dots \cup G_{i-1}} x^{r_p} \right)
      \left( \sum_{B_q \in G_i} x^{b_q} \right)
  \]
  To compute these contributions, iterate over each group in order, maintaining
  the coefficient form of the polynomial representing all red points in groups
  processed thus far. This corresponds to the first multiplicand on the right
  hand side.  We can quickly construct the second multiplicand directly from
  the elements in this group. Note that the degree of both multiplicands does
  not exceed $v$, and that the coefficients of the product do not exceed $n^2$.
  Hence, we can compute the product of these two polynomials in $O(v \log n)$
  time by Lemma~\ref{lem:fast-mult}, so we can compute the cross-group
  contributions in $O(\frac{n}{t} v \log n)$ time.

  Combining these parts with an appropriate choice of $t$ gives the required result.
  % Combining these parts and accounting for our initial sort, we obtain a
  % running time of $O(n \log n + nt + \frac{n}{t} v \log n)$. When $n \geq
  % \sqrt{v \log n}$, a choice of $t = \sqrt{v \log n}$ gives the desired time
  % complexity. Otherwise, $n < \sqrt{v \log n}$ and we can simply compare every
  % pair of points in $O(n^2) = O(n \sqrt{v \log n})$ time, which also fits in
  % the required time complexity.
\end{proof}

\begin{theorem}
  \label{bw-multi-d}

  \sloppy
  There is an algorithm that solves \textsc{RedBluePolynomial} in $2^{O(d)}
  n^{1+\varepsilon} \sqrt{v}$ time on a log-RAM, for every $\varepsilon > 0$.
\end{theorem}

%\begin{proof}[Proof (Algorithm~\ref{bw-multi-d})]
\begin{proof}
  When $d = 1$, we use Lemma~\ref{bw-one-d}. Otherwise, we will use the divide-and-conquer method
  of Bentley \cite{Bentley1980} to reduce the problem to smaller dimensions.

  First, combine the red and blue points into one list and apply
  divide-and-conquer as follows. Let $x_m$ be the median value among the first
  coordinate of all points. This can be found in $O(n)$ time \cite{Blum1973}. We
  divide the list into two halves as follows. First assign those points with
  first coordinate less than $x_m$ into the first half, and those with first
  coordinate greater than $x_m$ into the second half. Among those with first
  coordinate precisely $x_m$, assign blue points to the first half until the
  first half has $\frac{n}{2}$ points. Assign the remaining points to the
  second half. This assignment can be done in $O(n)$ time and has the property
  that if $R_p < B_q$, then either both points belong to the same half, or they
  belong to the first and second half, respectively.
  
  Next, recursively compute the contribution of both groups to the final
  polynomial.  The remaining pairs that may contribute terms to the result must
  have a red point in the first half, and a blue point in the second half.
  Since the ordering guarantees that all points in the first half have a first
  coordinate no greater than those in the second half, we project the red
  points in the first half together with the blue points in the second half
  onto a $(d-1)$-dimensional space by simply ignoring the first coordinate of
  each point.  We then solve \textsc{RedBluePolynomial} for this set of points
  in $d-1$ dimensions recursively. 

  Let $T(n, d)$ be the time complexity of this algorithm. Then, $T(n, d)$ satisfies the recurrence:
  \begin{align*}
    T(n, d) &= 2 T(n/2, d) + T(n, d-1) + O(n) \\
    T(n, 1) &= O(n \sqrt{v \log n} + n \log n) \\
    T(1, d) &= O(1).
  \end{align*}
  Bringmann et al. \cite{Bringmann2018}, using the results of Monier
  \cite{Monier1980}, showed that the recurrence
  \begin{align*}
    Q(n, d) &= 2 Q(n/2, d) + Q(n, d-1) + O(n) \\
    Q(n, 2) &= O(n \log n) \\
    Q(1, d) &= O(1)
  \end{align*}
  satisfies $Q(n, d) = 2^{O(d)} n^{1 + \varepsilon}$, for every $\varepsilon >
  0$.
  
  Observing that $T(n, d) = O( Q(n, d+1) \sqrt{v \log n} )$ and that
  $\sqrt{\log n} = O(n^{\varepsilon^{\ast}})$ for any $\varepsilon^{\ast} > 0$,
  completes the proof.
\end{proof}
\noindent

An analysis of the algorithm we have described in this section gives \autoref{thm:dist-pref}.

\thmdistpref*

\begin{proof}
  % We adopt the method described in this section, and consider its time complexity.

  To find the contribution of pairs in
  $(A \setminus S) \times (V \setminus A)$, we solve $\abs{S} \leq k + 1$ instances of
  \textsc{RedBluePolynomial} in $\abs{S} - 1 \leq k$ dimensions, using the result
  of Theorem~\ref{bw-multi-d}. As our algorithm performs divide-and-conquer over
  the nodes of the tree decomposition, each vertex induces the creation of a
  point in $O((k+1) \log (kn)) = O(k \log n)$ instances of \textsc{RedBluePolynomial}.
  Hence, since the time complexity of Theorem~\ref{bw-multi-d} is superadditive
  with respect to $n$, the total running time over all instances of
  \textsc{RedBluePolynomial} is $2^{O(k)} n^{1 + \varepsilon} \log n \sqrt{p} =
  2^{O(k)} n^{1 + \varepsilon'} \sqrt{p}$ for any $\varepsilon' > 0$.

  Since we are working on a (nice) tree decomposition with $O(kn)$ nodes, the
  running time of finding an appropriate dividing edge in the tree, and
  performing $k$ Dijkstra's per instance are negligible compared to that of
  solving our instances of \textsc{RedBluePolynomial}.
  The result follows from the fact that $k = O(\text{tw}(G))$.
\end{proof}

This result can easily be extended to directed graphs, and
graphs with bounded edge weights, with some modifications, and a suitable choice of $p$.
On graphs with unit weight edges, setting $p = n-1$ determines the entire distance distribution.

\begin{corollary}
  The distance distribution of an undirected graph $G$ with $n$ vertices, edges of
  unit weight and treewidth $\text{tw}(G)$ can be computed in $2^{O(\text{tw}(G))} n^{3/2 +
  \varepsilon} $ time on a log-RAM.%, for any $\varepsilon > 0$.
\end{corollary}

\section{Conclusion}

%\jl{Need to rewrite some parts of conclusion in light of MinH}
%\jl{Most of this is new.}
We have provided a general method to solve \textsc{MinH} on trees in
subexponential time and polynomial space, whenever $H$ is additive, balanced on trees, and computable in polynomial
time.
We used this to give a $2^{O((n \log n)^{5/6})}$ time, polynomial space algorithm for \textsc{MinIGL}, by proving that IGL is balanced on trees.
Our proof ideas can be used to show that other measures (such as the Wiener index), are also balanced on trees.

% that ask to minimize other measures on a forest, by deleting $k$ vertices
% from a tree with $n$ vertices. Specifically, the method applies when:
% \begin{enumerate}
%   \item the value of the measure on a forest can be computed in polynomial time,
%     and is equal to the sum of the value of the measure on each of its trees;
%     and
%   \item there exists some constant $c > 0$, such that, for any instance, there
%     exists an optimal solution that is
%     $O\left(\left(\frac{n}{k}\right)^c\right)$-trimming.
% \end{enumerate}

%\jl{This paragraph seems dodgy, not sure if it makes much sense to talk about
%extra factors.} %SG: stated entire running time first
%For graphs with treewidth $k$,
For graphs with treewidth $k$, we have shown that in $2^{O(k)} n^{3/2
+\varepsilon} $ time, one can compute the entire distance distribution of the
input graph. Compared to the $O(kn^2)$ time algorithm for computing APSP
\cite{Planken2012}, our
dependence on $n$ is a factor of $O(\sqrt{n})$ less, though our dependence on
$k$ is exponential. Our algorithm is a $O(\sqrt{n})$ factor slower than the
current best-known $2^{O(k)} n^{1 + \varepsilon}$ time algorithm for diameter
\cite{Abboud2016}. For graphs with diameter $O(n^{\varepsilon'})$ for all $\varepsilon' > 0$, including graphs with
polylogarithmic diameter, the extra
factor becomes $O(n^\varepsilon)$ for any $\varepsilon > 0$, when compared to
the current best-known $2^{O(k)} n$ time algorithm for diameter \cite{Husfeldt2017} in this setting. This might be expected, as the distance
distribution implies the diameter, and is implied by the APSP, but we find it
somewhat surprising that the distance distribution can be computed faster than APSP on graphs with small treewidth.
% changed this slightly

Our results can be immediately applied to compute any measure of a graph that
is a function of the distance distribution. However, they are
difficult to adapt to measures that compute properties of individual vertices in
the graph, as we exploit properties exclusive to counting pairs that are
certain distances apart, without expressly considering which vertices belong to
such pairs. In particular, this means that our results are unlikely to directly
provide further insight into the efficient computation of related measures, such
as the task of computing closeness centrality \cite{Bavelas1950a} of every
vertex in a given graph.
\bibliography{references}

% \appendix
% 
% \clearpage
% 
%\section{Paper count constraint}
%
%``Authors are invited to submit an extended abstract of no more than 12 pages, excluding references presenting original research on the theory of computer science.''
%
%``Technical details necessary for a proper scientific evaluation of a submission must be included in the 12-page submission or in a clearly labelled appendix, to be consulted at the discretion of program committee members.''
%
%Let us use the reviewer comments to decide which parts to move to the appendix.
%Easy proofs should definitely be moved, but we might need to move larger chunks.

\end{document}